\documentclass[
  a4paper,
  USenglish,
  cleveref,
  autoref,
  thm-restate,
]{lipics-v2021}

\usepackage[most]{tcolorbox}
\tcbuselibrary{breakable}
\usepackage{longtable}
\newtcolorbox{mainquestionbox}{
  colback=gray!5,
  colframe=black!75,
  width=\textwidth,
  arc=3mm,
  boxrule=0.8pt,
  left=6pt,
  right=6pt,
  top=6pt,
  bottom=6pt
}
\usepackage{graphicx}
\usepackage{amsmath,amsfonts,amsthm}
\usepackage{tcolorbox}
\usepackage{todonotes}
\setuptodonotes{inline}
\usepackage{xspace}
\usepackage{booktabs,tabularx,array}
\usepackage{thmtools}
\usepackage{thm-restate}
\usepackage{mdframed}

\newtcolorbox{grayframe}{
  breakable,
  colback=gray!10,
  colframe=black,
  boxrule=0.6pt,
  arc=2mm,
  left=3mm,
  right=3mm,
  top=2mm,
  bottom=2mm
}

\crefname{theorem}{Theorem}{Theorems}
\crefname{lemma}{Lemma}{Lemmas}
\crefname{observation}{Observation}{Observations}
\crefname{claim}{Claim}{Claims}

\newcommand{\R}{\mathcal{R}}

\newcommand{\nat}{\mathbb{N}}
\newcommand{\dist}{\mathsf{dist}}
\newcommand{\upperboud}{\mathcal{L}}
\newcommand{\bellowPara}{\zeta}
\newcommand{\Oh}{\mathcal{O}}

\newcommand{\CMPlong}{\textsc{Transient Multiagent Pathfinding}\xspace}
\newcommand{\CMPshort}{\textsc{Transient Multiagent Pathfinding}\xspace}

\DeclareMathOperator{\operatorClassNP}{NP}
\newcommand{\classNP}{\ensuremath{\operatorClassNP}\xspace}

\DeclareMathOperator{\operatorClassCoNP}{coNP}
\newcommand{\classCoNP}{\ensuremath{\operatorClassCoNP}}
\DeclareMathOperator{\operatorClassFPT}{FPT\xspace}
\newcommand{\classFPT}{\ensuremath{\operatorClassFPT}\xspace}
\DeclareMathOperator{\operatorClassW}{W}
\newcommand{\classW}[1]{\ensuremath{\operatorClassW[#1]}}

\DeclareMathOperator{\operatorClassParaNP}{Para-NP\xspace}
\newcommand{\classParaNP}{\ensuremath{\operatorClassParaNP}\xspace}
\DeclareMathOperator{\operatorClassXP}{XP\xspace}
\newcommand{\classXP}{\ensuremath{\operatorClassXP}\xspace}

\pdfoutput=1 
\hideLIPIcs  
\nolinenumbers

\title{Routing Multiple Agents Below the Sum of Distances}

\author{Matthias Bentert}{TU Berlin, Germany}{bentert@tu-berlin.de}{https://orcid.org/0009-0009-0705-972X}{Supported by the German Research Foundation (DFG) under ReNO-2 project (grant no. 1-5003036).}

\author{Eduard Eiben}{Royal Holloway, University of London, Egham, United Kingdom}{eduard.eiben@rhul.ac.uk}
    {https://orcid.org/0000-0003-2628-3435}
    {Engineering and Physical Sciences Research Council (EPSRC) grant UKRI4530.
    }

\author{Fedor V. Fomin}{University of Bergen, Norway}{fedor.fomin@uib.no}{https://orcid.org/0000-0003-1955-4612}
{Supported by the Research Council of Norway under BWCA project, reference 314528,  and European Research Council (ERC) via grant NewPC, reference 101199930.}

\author{Petr A. Golovach}{University of Bergen, Norway}{petr.golovach@uib.no}{https://orcid.org/0000-0002-2619-2990}
{Supported by the Research Council of Norway under BWCA  (grant no.~314528) and Extreme-Algorithms (grant no~355137) projects.}

\authorrunning{M. Bentert, E. Eiben, F. V. Fomin, P. A. Golovach} 

\Copyright{Matthias Bentert, Eduard Eiben, Fedor V. Fomin, Petr A. Golovach} 



\begin{document}

\maketitle

\begin{abstract}
We study \CMPlong, a variant of the classical \textsc{Multi-Agent Pathfinding} problem in which a set of agents must be routed without collisions from designated start vertices to designated destination vertices in a graph.

We analyze the problem within the above-and-below-guarantee paradigm of parameterized complexity. 
In particular, we consider the natural upper bound $\upperboud$, given by the sum of the shortest-path distances between pairs of agents' terminals (corresponding to sequential routing of the agents). 
The parameterization is given by the gap $\bellowPara = \upperboud - \lambda$ between this bound and the target makespan~$\lambda$, together with the number~$k$ of agents.  
Our main result establishes fixed-parameter tractability for the combined parameter~$k + \bellowPara$.

Matching lower bounds show that parameterization by~$k$ alone is \classW1-hard, and that parameterization by~$\bellowPara$ alone is \classW1-hard when terminals are not required to be distinct. 
On the positive side, if all terminals are distinct, the problem becomes fixed-parameter tractable when parameterized solely by~$\bellowPara$. 
Finally, we show that \CMPlong{} is unlikely to admit a polynomial kernel when parameterized by~$k + \bellowPara$.  
Together, our results provide an almost complete characterization of the parameterized complexity landscape of the problem for the considered parameters.
\end{abstract}

\maketitle 

 
\section{Introduction}\label{sec:introduction}
 
Gnomes are not known for their sociability.
They strongly prefer solitude
and take great care never to encounter one another.
Unfortunately for them, they must occasionally traverse a shared network of  tunnels in order to reach their individual destinations.

We model the tunnel system as a graph, where vertices represent chambers and junctions and edges represent narrow passages between them.
Each gnome is assigned a start chamber and a destination chamber.
Time advances in discrete steps, and at each step a gnome may either move to an
adjacent chamber or remain still.
Any encounter between two gnomes is considered disastrous: no two gnomes may
occupy the same chamber at the same time, nor may they traverse the same tunnel
simultaneously, not even in opposite directions.
Once a gnome reaches its destination, it immediately disappears into its private
quarters and no longer interacts with the system.

A simple scenario that allows all gnomes to reach their destinations is the
following.
The gnomes first agree on an order and then move one after another.
The first gnome traverses a shortest path to its destination; once it arrives
and disappears, the second gnome starts moving, and so on.
This straightforward strategy is easy to compute and yields a schedule whose
makespan, that is, the total time needed to move the gnomes to their destinations, is upper bounded by $1$ plus the sum of the lengths of the shortest paths of all
gnomes.
Achieving even a slightly smarter strategy  quickly becomes nontrivial:
Avoiding encounters may require gnomes to wait inside chambers, to take detours,
or to carefully interleave their movements.

\begin{mainquestionbox}
The main question we address in this work is whether it is possible 
to obtain a schedule with a smaller makespan and how difficult 
this task is computationally.
\end{mainquestionbox}

Of course, while the life of a gnome is certainly an interesting subject, the main question we address arises from its connections to more practical, real-life problems.

\medskip 
Routing multiple agents from designated start vertices to designated destination vertices without collisions is a fundamental problem in multi-agent planning. 
The problem is strongly motivated by numerous real-world applications, including automated warehouse systems~\cite{0001TKDKK21,kiva}, traffic management and control~\cite{MorrisPLMMKK16}, and robotics~\cite{VelosoBCR15}. A well-known example is the coordination of large fleets of mobile robots in Amazon fulfillment centers, where systems such as Kiva rely on efficient and reliable multi-robot path planning to transport inventory pods while avoiding collisions~\cite{kiva}. From a theoretical perspective, almost all variants of the problem are \classNP-hard~\cite{nphard2,nphard1,lavalle}.

 The range of applications such as automated warehouses, aircraft towing, UAV traffic management, or drone delivery  has led to the development of numerous problem variants (see, e.g.,~\cite{SternSFK0WLA0KB19,Varambally0K22} for surveys).

The problem is commonly formalized as \textsc{Multiagent Pathfinding} (\textsc{MAPF}), where the environment is modeled as an $n$-vertex graph and each of the $k$ agents has a specified start vertex $s_i$ and goal vertex $t_i$. At each discrete time step, an agent either waits or moves along an edge to a neighboring vertex; no vertex or edge may be used by more than one agent in the same time step. In the standard formulation of \textsc{MAPF}, the objective is to find a schedule that routes all agents to their goals while minimizing the \emph{makespan}~$\lambda$ (the arrival time of the last agent)~\cite{banfi,DeligkasEGK025,demaine,EibenGK23,FioravantesKKMO24}.
Another popular measure is the total distance traveled by the agents~\cite{DeligkasEGKLR26,icalp,GeftHalperin,PapadimitriouRST94}. 
Let us remark that the total distance travel by the agents is always at least the sum of the lengths of the shortest paths between all terminals.

We focus on algorithms with provable guarantees for computing schedules of optimal makespan. The problem has been extensively studied and is computationally intractable in the classical sense: it is NP-hard not only on general graphs, but also on solid grid graphs~\cite{demaine} and even on trees~\cite{FioravantesKKMO24}. Given this intractability, numerous heuristic and algorithmic approaches have been proposed, including SAT-based methods~\cite{surynek2010optimization},  and $A^*$-based algorithms~\cite{heuristic2}; see also the surveys~\cite{SternSFK0WLA0KB19, Varambally0K22}.

There are two common assumptions regarding how agents behave upon reaching their targets; see~\cite{SternSFK0WLA0KB19}.
Under the first assumption, an agent remains at its target vertex until all agents have reached their respective targets.
Under the second assumption, agents disappear immediately after reaching their targets.
Our results concern the latter assumption.
We adopt the model in which agents enter the graph at their start times at $s_i$, traverse along edges, and immediately disappear upon reaching $t_i$; in particular, $t_i$ is not occupied after arrival.

In this work, we adopt the framework of parameterized complexity~\cite{DowneyFellows13,CyganFKLMPPS15}.
Although the computational complexity of many fundamental coordinated motion planning variants was established several decades ago~\cite{PapadimitriouRST94}, recent years have witnessed renewed interest in these problems through the lens of parameterized complexity~\cite{halprinunlabeled,DeligkasEGK025,DemaineHM14,EibenGK23,EGKS25SoCG,FioravantesKKMO24,FioKKMO25}.

Most of this line of work focuses on natural parameterizations such as the number of agents, the makespan, or structural parameters of the input graph, for example its treewidth.
We take a different perspective here.
In particular, we follow the above-or-below-guarantee paradigm, a popular approach in parameterized complexity. 
The term {above (below) guarantee parameterization} was introduced by Mahajan and Raman~\cite{MahajanRS09}. This concept has proved fruitful for studying numerous problems in parameterized complexity  ~\cite{AlonGKSY10,BezakovaCDF17,CrowstonJMPRS13,FominGLPSZ21,GargP16,GutinRSY07,LokshtanovNRRS14}, see also the survey \cite{GutinMnich22}.

Consider the simple sequential protocol that routes the agents one after another along shortest $s_i$-$t_i$-paths.  
It is not difficult to show (see \Cref{lem:upperbound})
that this protocol guarantees completion within~$\lambda \;\le\; \upperboud_{\R,G}
:= 1 \;+\; \sum_{i\in[k]} \dist_G(s_i,t_i)$ time steps.

This bound is tight: there exist instances of the problem  that require exactly
$\upperboud_{\R,G}$ rounds.
For example, let $G$ be a path with end-vertices $u$ and $v$, and consider two agents
with terminals $s_1 = t_2 = u$ and $s_2 = t_1 = v$.
In this case, neither agent can start moving before the other has reached its
destination and disappeared; otherwise, they would block each other along the path.
Consequently, one agent must wait until the other completes its traversal,
incurring at least one additional round between their movements.
Therefore, any feasible schedule has makespan at least~$1 + \dist_G(s_1,t_1) + \dist_G(s_2,t_2)
= \upperboud_{\R,G}$.
Hence, the parameter $\bellowPara_{\R,G} = \upperboud_{\R,G} - \lambda$
quantifies how much tighter a schedule we obtain compared to the
(sometimes pessimistic) upper bound $\upperboud_{\R,G}$.
In the terminology of parameterized complexity, parameterizing by
$\bellowPara_{\R,G}$ corresponds to \emph{parameterization below the trivial upper bound}
$\upperboud_{\R,G}$.

Before proceeding with an overview of our results, let us give an example of a possible scheduling of the agents, postponing the formal definitions to \Cref{sec:formaldefinition}.

\medskip
      
      \begin{tcolorbox}[breakable,
  title={Example},
  colback=white,
  colframe=black,
  fonttitle=\bfseries,
  sharp corners,
  boxsep=1mm,
  left=1mm,
  right=1mm,
  top=1mm,
  bottom=1mm
]

Consider the following graph with two agents~$R_1$ and~$R_2$, where~$R_1$ needs to move from $a$ to $d$ and~$R_2$ from $d$ to $a$.

\begin{center}
\begin{tikzpicture}[
  every node/.style={circle,draw,inner sep=1.5pt,font=\small},
  scale=1
]
  \node (a) at (0,0) {$a$};
  \node (b) at (1.6,0) {$b$};
  \node (c) at (3.2,0) {$c$};
  \node (d) at (4.8,0) {$d$};
  \node (e) at (1.6,1.3) {$e$};

  \draw (a) -- (b) -- (c) -- (d);
  \draw (b) -- (e);

  \draw[-,thick] (a) -- (b);
  \draw[-,thick] (b) -- (e);
  \draw[-,thick] (e) -- (b);
  \draw[-,thick] (b) -- (c);
  \draw[-,thick] (c) -- (d);

  \draw[-,thick,dashed] (d) -- (c);
  \draw[-,thick,dashed] (c) -- (b);
  \draw[-,thick,dashed] (b) -- (a);

  \node[draw=none,rectangle,font=\footnotesize] at (2.5,-0.9) {%
    $R_1$ \tikz{\draw[->,thick] (0,0)--(0.7,0);} 
    \quad \quad \quad \quad  \quad \quad \quad \quad \quad
    \tikz{\draw[<-,thick,dashed] (0,0)--(0.7,0);} $R_2$
  };
\end{tikzpicture}
\end{center}

The upper bound $\upperboud_{\R,G}$ in this case is \[1+\sum_{i\in[2]}\dist_G(s_i,t_i)=1+  \dist_G(a,d) +\dist_G(d,a)=1+3+3=7.\]
If both agents traverse the shortest path $a$--$b$--$c$--$d$ starting at time~$0$ in opposite directions, then at some time step they traverse the same edge in opposite directions. These routes are \emph{conflicting}.
A valid (but not optimal) schedule is when first $R_1$ moves from $a$ to $d$, and then $R_2$ moves from $d$ to $a$. That is, 
\begin{center}
\footnotesize
\begin{tabular}{c|cc}
\toprule
time step $x$ & $R_1$ & $R_2$ \\
\midrule
$0$ & $a$ & --- \\
$1$ & $b$ & --- \\
$2$ & $c$ & --- \\
$3$ & $d$ (disappears) & ---  \\
$4$ & --- & $d$ \\
$5$ & ---  & $c$ \\
$6$ & --- & $b$   \\
$7$ & --- & $a$ (disappears) \\
\bottomrule
\end{tabular}
\end{center}
The makespan $\lambda$ in this case is equal to $7$,  the moment agent $R_2$ reaches $a$. It is also equal to the upper bound $\upperboud_{\R,G}.$

\medskip 
Consider another, now optimal, scheduling, where  $R_1$ use the side vertex $e$ as a buffer:

\begin{center}
\footnotesize
\begin{tabular}{c|cc}
\toprule
time step $x$ & $R_1$ & $R_2$ \\
\midrule
$0$ & $a$ & $d$ \\
$1$ & $b$ & $c$ \\
$2$ & $e$ & $b$ \\
$3$ & $b$ & $a$ (disappears) \\
$4$ & $c$ & --- \\
$5$ & $d$ (disappears) & --- \\
\bottomrule
\end{tabular}
\end{center}

At no time step do the agents occupy the same vertex or traverse the same edge simultaneously. Hence, the routes are non-conflicting.
Agent $R_1$ finishes at time $5$, agent $R_2$ at time~$3$, so the makespan is
$
\lambda = 5, 
$ improving over the upper  $\upperboud_{\R,G}$ by 
$\bellowPara_{\R,G} = \upperboud_{\R,G} - \lambda=2$.

\end{tcolorbox}

\noindent\textbf{Our contribution.}
We consider the following problem.

\begin{quote}
\textsc{\CMPlong}

\smallskip
\noindent\textbf{Input:} A connected graph $G$, integers $k$
and~$\lambda$, and a set $\R=\{R_i=(s_i,t_i)\mid i\in[k]\}$
of agents.

\smallskip
\noindent\textbf{Task:} Is there a schedule for $\R$ of makespan
at most $\lambda$?
\end{quote}

Our main result shows that the problem is fixed-parameter tractable with respect to the combined parameter $k + \bellowPara$, where $k$ is the number of agents and $\bellowPara$ measures the improvement over the upper bound $\upperboud$.

\begin{restatable}{theorem}{FPTBelowTheorem}\label{thm:fptbelow}
\CMPshort can be solved in $2^{\Oh(k^2\bellowPara)} \cdot n^{\Oh(1)}$ time
on instances $(G,k,\R=\{R_i=(s_i,t_i) \mid i \in [k]\},\lambda)$ with
$\lambda=\upperboud_{\R,G}-\bellowPara$.
\end{restatable}
\medskip\noindent\emph{Overview of the proof of \Cref{thm:fptbelow}}. 
The proof  is based on structural properties of shortest paths. The starting point is the fact (see \Cref{lem:fptlambda}) that \CMPshort can be solved in $2^{\Oh(k\lambda)}\cdot n^{\Oh(1)}$ time. Thus, \Cref{thm:fptbelow} holds if $\lambda=\upperboud_{\R,G}-\bellowPara=\Oh(k\bellowPara)$. In particular, this is fulfilled if $\max_{i\in[k]}\dist_G(s_i,t_i)=\Oh(\bellowPara)$. We prove that solving \CMPshort can be reduced to this special case when the length of shortest $s_i$-$t_i$-paths is  $\Oh(\bellowPara)$. For this, we consider an instance where there is $i\in[k]$ such that the distance between $s_i$ and $t_i$ is sufficiently large, and either demonstrate that this is a no-instance or construct a routing with makespan at most $\lambda$. Moreover, to rule out the no-instances, it is sufficient to verify whether there is $i\in[k]$ with $\dist_G(s_i,t_i)>\lambda$ or there are distinct $i,j\in[k]$ such that the agents $R_i$ and $R_j$ cannot be routed within the makespan $\lambda$. The latter can be checked by an \classXP algorithm solving \CMPshort in  
$(kn)^{\Oh(k)}$ time (see~\Cref{lem:XP}). In all other cases, we are able to construct a routing schedule of the makespan at most $\lambda$.

We briefly sketch the main ideas behind constructing a routing schedule in the case when there is a pair of faraway terminals. The proof is technical and is obtained by a sequence of lemmas  establishing sufficient conditions for the existence of a feasible schedule. 

Given two terminal pairs $(s_i,t_i)$ and $(s_j,t_j)$, we consider \emph{laminar} shortest $s_i$-$t_i$- and $s_j$-$t_j$-paths $P_i$ and $P_j$, respectively, that is, shortest paths whose intersection is either empty or a path. We distinguish \emph{codirectional} and \emph{opposite} paths. Paths $P_i$ and $P_j$ are codirectional if the agents traverse the common subpath in the same direction following $P_i$ and $P_j$, and the paths are opposite, otherwise (see~\Cref{fig:laminar}). We note that the agents can follow codirectional shortest paths almost simultaneously. This implies that if there are codirectional shortest paths $P_i$ and $P_j$ of length $\Omega(\bellowPara)$, then we can route the agents $R_i$ and $R_j$ saving considerable time, and this is sufficient for the existence of a routing of the makespan at most $\lambda$.
We also observe that in some cases, opposite paths of length $\Omega(\bellowPara)$ can be used to save time. In particular, this happens if  the common part of the paths is sufficiently far from one of the end-points of the paths.

If the above cases do not apply, we observe that if there are three terminal pairs with terminals at distance $\Omega(\bellowPara)$, we can find codirectional shortest paths for two pairs. Thus, the analysis boils down to the cases when the terminals of either one or two pairs are at distance at least $c\cdot\bellowPara$ for a specific constant $c>0$, and for all other pairs $(s_i,t_i)$, the distance between $s_i$ and $t_i$ is upper bounded by $c'\cdot\bellowPara$ for another constant $c'<c$.
If there is exactly one shortest path of length at least $c\cdot \bellowPara$, then we can route the corresponding agent $R_i$ along a shortest path. Simultaneously with $R_i$, we can route other agents and save time. 

The case of two pairs of faraway terminals is the most challenging. Here, we need some additional structural results about paths in $G$. Very roughly, the idea is to route two agents~$R_i$ and~$R_j$ with $\dist_G(s_i,t_i),\dist_G(s_j,t_j)\geq c\cdot\bellowPara$ in the optimum way (which is proved to be very special after excluding some cases), and route the other agents while $R_i$ and $R_j$ follow their routes.
This concludes the proof overview.

\medskip
Of course, a natural question is whether \Cref{thm:fptbelow} is tight. In other words, whether the combined parameterization is necessary. We complement our theorem with lower bounds showing that this is indeed the case. First, we observe (\Cref{prop:W1k} and \Cref{corrolary:w1-hard}) that \CMPshort{} is \classW1-hard when parameterized by~$k$, even if the input graph is subcubic. We also prove (\Cref{thm:W1b}) that \CMPshort{} is \classW1-hard when parameterized by $\bellowPara$. Therefore, fixed-parameter tractability with respect to only one of these parameters is highly unlikely.
Regarding single-parameter parameterizations, we observe (\Cref{lem:XP}) that for fixed $k$, \CMPshort{} is solvable in polynomial time $(kn)^{\Oh(k)}$. That is, it belongs to \classXP  when parameterized by $k$. Whether the problem is \classParaNP-complete or in \classXP  when parameterized by $\bellowPara$ remains an interesting open question.

Having established an \classFPT result, a natural next question concerns the existence of a polynomial kernel. We show in \Cref{thm:nopolykernel} that \CMPshort{} parameterized by~$k + \bellowPara$ does not admit a polynomial kernel. Combined with the above lower bounds, \Cref{thm:fptbelow} yields an almost complete characterization of the parameterized complexity of \CMPshort{}. We summarize these results in \Cref{tab:summary}.

\begin{table}[t]
\centering
\begin{tabular}{ll}
\toprule
\textbf{Parameter} & \textbf{Result} \\
\midrule
$k + \bellowPara$ 
& FPT in $2^{\Oh(k^2\bellowPara)}\cdot n^{\Oh(1)}$ (\Cref{thm:fptbelow}); no polynomial kernel (\Cref{thm:nopolykernel})  \\

$k$ 
& \classW1-hard (subcubic graphs) (\Cref{corrolary:w1-hard});  In \classXP  (\Cref{lem:XP}) \\

$\bellowPara$ 
& \classW1-hard (\Cref{thm:W1b});  \classParaNP\ vs.\ \classXP is open\\
\bottomrule
\end{tabular}
\caption{Parameterized complexity of \CMPshort.}
\label{tab:summary}
\end{table}

Although our lower bounds exclude the existence of an \classFPT algorithm for \CMPshort{} parameterized solely by $\zeta$, we identify a natural  subclass that does admit such algorithms. In particular, when all terminals are distinct, the upper bound on $\lambda$ improves to $\upperboud_{\R,G}^* = \sum_{i \in [k]} \dist(s_i,t_i)-\lfloor\frac{k-1}{2}\rfloor$. We establish the following theorem.
\begin{restatable}{theorem}{FPTBelowTheoremTwo}\label{thm:fptbelow-two}
\CMPshort can be solved in $2^{\Oh(\bellowPara^3)}\cdot n^{\Oh(1)}$ time on instances $(G,k,\R=\{R_i=(s_i,t_i) \mid i \in [k]\},\lambda)$ with pairwise distinct terminals and~$\lambda=\upperboud_{\R,G}^*-\bellowPara$. 
\end{restatable}

\section{Preliminaries}\label{sec:formaldefinition}
\begin{table}[t]
\setlength{\tabcolsep}{4pt}
\footnotesize
\centering
\begin{tabular}{cl}
\toprule
\textbf{Symbol} & \textbf{Meaning} \\
\midrule
$\nat$ & Non-negative integers. \\
$G$ & Undirected input graph. \\
$V(G),\,E(G)$ & Vertex and edge sets of $G$. \\
$\dist_G(a,b)$ & Distance (the number of edges in the shortest path) between\\
& $a$ and $b$ in $G$. \\
$k$ & Number of agents. \\
$\R=\{R_1,\dots,R_k\}$ & Set of agents. \\
$R_i=(s_i,t_i)$ & Agent $i$ with start $s_i$ and destination $t_i$. \\
$s_i,\,t_i$ & Start/destination of $R_i$ ($s_i\neq t_i$). \\
$\{s_i\mid i\in[k]\}\cup\{t_i\mid i\in[k]\}$ & Set of  {terminals}. \\
$(W_i,x_i)$ & Route of $R_i$. \\
$W_i=(u_0,\ldots,u_{\delta_i})$ & Positions of $R_i$ along its route. \\
$x_i$ & Moment when  $R_i$ starts moving along $W_i$. \\
$\delta_i$ & Length of $W_i$ (steps). \\
$u_y$ & Position of $R_i$ at time $x_i+y$. \\
$\mathrm{free}/\mathrm{occ.}$ & Vertex status at a time step (no/some agent at~$v$). \\
$\Gamma$ & A schedule: pairwise non-conflicting routes. \\
$\lambda$ & Makespan of $\Gamma$ ($\max_i(x_i+\delta_i)$). \\
$\upperboud_{\R,G}$ (or $\upperboud$) & $1+\sum_{i\in[k]}\dist_G(s_i,t_i)$, the  upper bound on the makespan. \\
$\bellowPara_{\R,G}$  (or $\bellowPara$) & $\upperboud_{\R,G}-\lambda$, the gain below the upper bound. \\
\bottomrule
\end{tabular}
\caption{Summary of notation.}
\label{tab:notation}
\end{table}

We use $\nat$ to denote the set of non-negative integers, and write $[p]$ and $[0,p]$ for the set~$\{1,\ldots,p\}$ and $\{0,\ldots,p\}$, respectively, 
where $p\in\nat$.
We refer to the textbook~\cite{Diestel} for the basic graph-theoretic notions and notation.
The sets of vertices of a graph graph $G$ are denoted by~$V(G)$ and $E(G)$, respectively. 
We write $P=v_0,\ldots,v_\ell$, where $v_0,\ldots,v_\ell\in V(G)$ are distinct vertices of $G$ such that $v_{i-1}v_i\in E(G)$ for all $i\in[\ell]$, to denote a \emph{path} of \emph{length} $\ell$ in $G$. The vertices $v_0$ and $v_\ell$ are the \emph{end-vertices} of $P$, and $v_1,\ldots,v_{\ell-1}$ are \emph{internal}. For a path $P$ with end-vertices $s$ and $t$, we say that $P$ is an $s$-$t$-path. 
For an undirected graph $G$ and two vertices $u,v\in V(G)$, we denote by $\dist_G(u,v)$ the distance, i.e., the length of the shortest path, between $u$ and $v$ in $G$. We drop the subscript, if the graph $G$ is clear from the context.  Throughout the paper, we use $n$ to denote the number of vertices of a graph if it does not create  confusion.

We are given an undirected graph~$G$ and a set~${\R=\{R_1, R_2, \ldots, R_k\}}$ of $k$ agents. Each~$R_i \in \R$ has a starting vertex $s_i$ and a destination vertex $t_i$ in $V(G)$ such that $s_i\neq t_i$.
We refer to the elements in the set~$\{s_i\mid i\in [k]\}\cup \{t_i\mid i\in [k]\}$ as \emph{terminals}.
We use a discrete time frame $[0, t]$, $t \in \nat$, to reference the sequence of moves of the agents, and in each time step $x \in [0, t]$, every agent~$R_i\in \R$ performs one of the following actions: 
\begin{itemize}
    \item $R_i$ enters the graph at $s_i$; 
    \item $R_i$ moves from its current position $v\in V(G)$ to a neighbor of $v$; 
    \item $R_i$ stays at its current position $v\in V(G)$; or 
    \item $R_i$ disappears from the graph if its current position is $t_i$.
\end{itemize}

A \emph{route} for agent $R_i$ is a pair $(W_i, x_i)$, where  $W_i$ is 
a sequence $(u_0, \ldots, u_{\delta_i})$ of vertices in~$G$ and  such that
\begin{itemize}
    \item $u_0$ is the starting vertex $s_i$ of $R_i$; 
    \item $u_{\delta_i}$ is the destination vertex $t_i$ of $R_i$; 
    \item for each $y\in [\delta_i]$ it holds that $u_{y-1}=u_y$ or $u_{y-1}u_{y}$ is an edge in $G$, 
\end{itemize}
and $x_i\in \nat$. 
The vertex $u_y$ is the position of $R_i$ at time step $x_i+y$. Thus, $R_i$ appears at moment $x_i$ in $s_i$ and moves along $W_i$ from $s_i$ to $t_i$; the agent disappears when it reaches $t_i$.
A vertex $v \in V(G)$ is \emph{free} at time step $x \in [0,t]$ if no agent is located at $v$ at time step $x$; otherwise, $v$ is \emph{occupied}. 

For $1\leq i<j\leq k$, two routes $(W_i=(u_0, \ldots, u_{\delta_i}),x_i)$ and $(W_j=(v_0, \ldots, v_{\delta_j}),x_j)$ are \emph{conflicting} (or \emph{in conflict}) if there exists some time step $x$ such that $R_i$ and $R_j$ are either at the same vertex or traversing the same edge at the time step $x$. More formally, $x= x_i+y = x_j+z$ such that either $u_{y} = v_{z}$ or $u_{y-1}u_{y}=v_{z-1}v_{z}$. A \emph{schedule} $\Gamma$ for $\R$ is a set of pairwise non-conflicting routes $\{(W_i,x_i)~|~i \in [k]\}$, during a time interval $[0, \lambda]$ (that is, $\max_{i\in [k]}(x_i+\delta_i) = \lambda$). The integer $\lambda$ is called the \emph{makespan} of $\Gamma$. Using the introduced terminology, we formalize \CMPlong stated in \Cref{sec:introduction}.
The above notation is summarized in \Cref{tab:notation}.

For the introduction to Parameterized Complexity, we refer to the textbooks~\cite{CyganFKLMPPS15,DowneyFellows13}.
Here, we just informally  remind some basic notions. 
A \emph{parameterized problem} $L\subseteq\Sigma\times\nat$, where~$\Sigma^*$ is the set of strings over a finite alphabet $\Sigma$. Respectively, its instance is a pair~$(x,k)$, where~$x\in\Sigma^*$ and $k\in\nat$ is a \emph{parameter}.
A parameterized problem $L$ is  \emph{fixed-parameter tractable} (\classFPT) if there exists an algorithm solving $L$ in $f(k)\cdot|x|^{\Oh(1)}$ time, where $f$ is a computable function. Furthermore, $L$ is \emph{slicewise polynomial} (\classXP) if $L$ can be solved in $|x|^{g(k)}$ time for a computable function $g$. The classes \classFPT and \classXP consist of the fixed-parameter tractable
and slicewise polynomial parameterized problems, respectively. The standard way to exclude the existence of an \classFPT algorithm for a parameterized problem (up to reasonable complexity assumptions) is to show its hardness for the class \classW{1} (see~\cite{CyganFKLMPPS15,DowneyFellows13} for the definitions)
using a \emph{parameterized reduction}, that is, a reduction running in $f(k)\cdot|x|^{\Oh(1)}$ time for a computable $f$.

A \emph{kernelization} (or \emph{kernel}) for a parameterized problem $L$ is a polynomial-time algorithm that for any instance $(x,k)$ outputs an instance $(x',k')$ of $L$ such that (i)~$(x,k)\in L$ if and only if $(x',k')\in L$ and (ii)~$|x'|+k'\leq f(k)$ for a computable function $f$. A kernel is \emph{polynomial} if $f$ is polynomial. While a decidable parameterized problem is in \classFPT if and only if it admits a kernel, it is unlikely that all fixed-parameter tractable parameterized problems admit polynomial kernels. In particular, it is standard to exclude the existence of a polynomial kernel up to the assumption that \classCoNP$\not\subseteq$\classNP/{\sf poly} by providing a \emph{polynomial parameter transformation}, that is, a parametrized reduction where $f$ is polynomial, from the parameterized problem for which such a kernelization lower bound is already established.

\section{Basic observations}
\label{sec:basic}

In this section, we collect basic  observations on \CMPshort. 
First, we formally prove our upper bound for the minimum makespan and show some useful properties of schedules guaranteeing such a makespan.

\begin{lemma}\label{lem:upperbound} 
    Let $(G, k, \R, \lambda)$ be an instance of \CMPshort. 
    If $\lambda \ge 1+ \sum_{i\in[k]} \dist(s_i,t_i)$, then $(G, k, \R, \lambda)$ is a yes-instance.  Moreover, there is a schedule of makespan at most $\lambda$ such that for every $i\in[k]$,
    \begin{itemize}
    \item[(i)] the route of $R_i$ is $(P_i,x_i)$ where $P_i$ is an arbitrary shortest $s_i$-$t_i$-path, and
    \item[(ii)] for each internal vertex $v$ of $P_i$ and the moment at which $R_i$ is in $v$, $R_i$ is an only agent on $G$.   
    \end{itemize}
\end{lemma}

\begin{proof}
   We assume that the agents $R_1,\ldots,R_k$ are ordered in such a way that $s_i\neq t_{i-1}$ holds for all~$i\in[2,\ell]$, with $\ell\leq k$ chosen as large possible.
   We navigate the agents $R_1,\ldots,R_\ell$ sequentially on shortest paths, such that we start the route of the first agent at the moment $0$ and
   the route of agent $i\geq 2$ exactly at the same time when the $(i-1)$-th agent reaches its destination (that is at the time step $\sum_{j\in[i-1]} \dist_G(s_j,t_j)$). It is straightforward  that the routes in this (partial) schedule have no conflicts.   
   If $\ell=k$ then we have the schedule with the makespan~$\sum_{i\in[k]} \dist_G(s_i,t_i)$.
   Assume that this is not the case and $\ell<k$. 

   By the choice of $\ell$, we have that for every $i\in[\ell+1,k]$, $s_i=t_\ell$ as, otherwise, we would be able to extend the ordering $R_1,\ldots,R_\ell$ by appending one of the remaining agents. We also have that for every $i\in[\ell+1,k]$, $t_i=s_1$ because we would be able to insert the agent with $t_i\neq s_1$ for $i\in[\ell+1,k]$ in the beginning of the ordering, contradicting the choice of $\ell$. Note that this implies that $s_1\neq t_\ell$ as $s_i\neq t_i$ for all $i\in[k]$. 
   Then, we can extend the partial schedule already constructed for $R_1,\ldots,R_\ell$ as follows.  
   We set 
   $t=\sum_{i\in[\ell]}\dist_G(s_i,t_i)+1$, and then we navigate the agents $R_{\ell},\ldots,R_k$ using the same strategy as for the previous agent starting from the moment  $t$. More precisely, for each $i\in[\ell+1,k]$ we move $R_i$ from $s_i$ to~$t_i$ along the shortest path, and $R_i$ starts moving at the moment 
   $1+\sum_{j\in[i-1]}\dist_G(s_j,t_j)$.  Because $s_i=t_\ell$ and $t_i=s_1\neq s_i$ for each $i\in[\ell+1,k]$ and $R_{\ell}$ starts to move at the moment $t+1$, the constructed schedule is feasible, that is, the routes have no conflicts. Because the makespan is 
   $1+\sum_{i\in[k]}\dist_G(s_i,t_i)$, we have that $(G, k, \R, \lambda)$ is a yes-instance if~$\lambda \ge 1+ \sum_{i\in[k]} \dist(s_i,t_i)$. 
   It is straightforward to see that the constructed schedule satisfies requirements (i) and~(ii).
   This concludes the proof.   
\end{proof}

We note that the bound in \Cref{lem:upperbound} is tight. Let $G$ be a path with end-vertices $u$ and~$v$. 
We define $s_1=t_2=u$ and $s_2=t_1=v$. Then routing two agents $R_1=(s_1,t_1)$ and~$R_2=(s_2,t_2)$ demands $\dist_G(s_1,t_1)+\dist_G(s_2,t_2)+1$ steps.
From now on, let us denote by~$\upperboud_{\R, G}$ (or just~$\upperboud$ if $G$ and $\R$ are clear from the context) the value $\upperboud_{\R, G} = 1+ \sum_{i\in[k]} \dist(s_i,t_i)$.
Similarly, we denote by $\bellowPara_{\R, G}$ (or just $\bellowPara$) the value~$\bellowPara_{\R, G}=\upperboud_{\R, G} - \lambda$.

In \Cref{thm:fptbelow}, we exploit the facts that \CMPshort is in~\classXP when parameterized by $k$ and in \classFPT for the parameterization by $k$ and $\lambda$. 
To show these properties, we reduce  the problem to finding paths in auxiliary 
time-expansion graphs.

Let $G$ be a graph and let $\lambda\geq 0$ be an integer. We denote by $\mathcal{G}_\lambda$ the directed \emph{time-expansion} graph with the vertex set $\{(v,t)\mid v\in V(G)\text{ and }t\in[0,\lambda]\}$ where there is an arc from a vertex $(u,i)$ to a vertex~$(v,j)$ if and only if $j-i=1$ and $v\in N_G[u]$. In other words, the set of vertices consists of $\lambda+1$ copies of $V(G)$ forming layers, arcs join vertices of subsequent layers, and $(u,i)$ from the $i$-th layer is joined with arcs with the copies of the vertices $v$ of $G$ in the $(i+1)$-th layer that are either the same as $u$ or adjacent to $u$ in $G$. Notice that $\mathcal{G}_\lambda$ is a directed acyclic graph (DAG).
The construction of $\mathcal{G}_\lambda$ and the definition of the routes in \CMPshort imply the following observation.

\begin{observation}\label{obs:paths_temp}
An instance  $(G,k,\R=\{R_i=(s_i,t_i) \mid i \in [k]\},\lambda)$ of \CMPshort is a yes-instance if and only if there are integers ${p_1,\ldots,p_k\in[0,\lambda]}$ and~${q_1,\ldots,q_k\in[0,\lambda]}$ such that $\mathcal{G}_\lambda$ has vertex disjoint directed paths $P_1,\ldots,P_k$, where~$P_i$ is an $(s_i,p_i)$-$(t_i,q_i)$-path for every~$i\in[k]$,
with the additional property that for every~$h\in[\lambda]$ and and each edge $\{u,v\}$ of $G$, the arcs $((u,h-1),(v,h))$ and $((v,h-1),(u,h))$ are not used by two paths.
\end{observation}

\begin{figure}[t]
\centering
\scalebox{0.7}{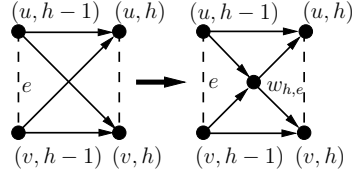}
\caption{Arc modifications in the construction of $\mathcal{G}_\lambda^*$. }
\label{fig:subdiv}
\end{figure}

Furthermore, we can augment $\mathcal{G}_\lambda$ by artificial terminals and modify the graph to encode the additional property to reduce \CMPshort to the \textsc{Disjoint Paths} problem. We remind the reader that the task of \textsc{Disjoint Paths} is, given a (directed) graph $G$ and $k$ pairs of terminal vertices $(s_i,t_i)$ for $i\in[k]$, to decide whether there are $k$ vertex disjoint paths joining these pairs of terminals.
Let $s_i,t_i\in V(G)$ be terminal vertices for~$i\in[k]$. We construct $\mathcal{G}_\lambda^*$ as follows:  
\begin{itemize}
\item construct $\mathcal{G}_\lambda$ and $2k$ vertices $s_i^*$ and $t_i^*$ for $i\in[k]$,
\item for each $i\in[k]$ and every $j\in[0,\lambda]$, construct arcs $(s_i^*,(s_i,j))$ and $((t_i,j),t_i^*)$,
\item for each $h\in[\lambda]$ and every edge $e=\{u,v\}$ of $G$, construct a vertex $w_{h,e}$ and replace the arcs $((u,h-1),(v,h))$ and $((v,h-1),(u,h))$ by four arcs 
${((u,h-1),w_{h,e})}$, and~${(w_{h,e},(v,h))},\allowbreak{((v,h-1),w_{h,e})},{(w_{h,e},(u,h))}$;
\end{itemize}
see \Cref{fig:subdiv} for the arc modifications in the construction.
Notice that $\mathcal{G}_\lambda^*$ is a DAG. By \Cref{obs:paths_temp}, we obtain the following claim.

\begin{observation}\label{obs:paths_aug}
An instance  $(G,k,\R=\{R_i=(s_i,t_i) \mid i \in [k]\},\lambda)$ of \CMPshort is a yes-instance if and only if $\mathcal{G}_\lambda^*$ has vertex disjoint directed paths $P_1,\ldots,P_k$, where $P_i$ is an $s_i^*$-$t_i^*$-path for every $i\in[k]$.
\end{observation}

\begin{proof}
To see the equivalence between the conditions in \Cref{obs:paths_temp} and \Cref{obs:paths_aug}, note that for  vertex disjoint paths $s_i$-$t_i$-paths in $\mathcal{G}_\lambda^*$, it can be assumed that they are induced. Then 
for every $h\in[\lambda]$ and each $e=\{u,v\}\in E(G)$, $((v,h-1),w_{h,e},(v,h))$ is not a subpath of any path. Otherwise, we can replace 
$((v,h-1),w_{h,e},(v,h))$ by $((v,h-1),(v,h))$. 
To conclude the proof, it is sufficient to observe that 
for every $h\in[\lambda]$ and each $e=\{u,v\}\in E(G)$, $((u,h-1),w_{h,e},(v,h))$ and $((v,h-1),w_{h,e},(u,h))$ are not disjoint and, therefore, cannot be subpaths of vertex disjoint paths in $\mathcal{G}_\lambda^*$.
\end{proof}

Because \textsc{Disjoint Paths} can be solved in $n^{\Oh(k)}$ time for $k$ pairs of terminals on DAGs by the result of Fortune, Hopcroft, and Wyllie~\cite{FortuneHW80}, by \Cref{obs:paths_aug}, we have the following. 

\begin{lemma}\label{lem:XP}
\CMPshort is  solvable in time $(kn)^{\Oh(k)}$.     
\end{lemma}

\begin{proof}
Let $I=(G,k,\R=\{R_i=(s_i,t_i) \mid i \in [k]\},\lambda)$ be an instance of \CMPshort where $G$ is a graph with $n$ vertices and $m$ edges. If ${\lambda\geq \upperboud_{\R, G}}$, then $I$ is a yes-instance by \Cref{lem:upperbound}. Otherwise, we construct $\mathcal{G}_\lambda^*$. As the distance between any pair of terminals is at most $n-1$, $\lambda\leq (n-1)k$.
Hence, we have that~$\mathcal{G}_\lambda^*$ has at most~${(\lambda+1)n+m\lambda+2k=((n-1)k+1)n+m(n-1)k+2k}$~vertices. Then, we call the algorithm by Fortune et al.~~\cite{FortuneHW80} to solve \textsc{Disjoint Path} on $\mathcal{G}_\lambda^*$ with the pairs~$(s_i^*,t_i^*)$ for~$i\in[k]$ of terminals in $(nk)^{\Oh(k)}$ time. By \Cref{obs:paths_aug}, solving \textsc{Disjoint Paths} for the constructed instance is equivalent to solving \CMPshort for $I$. This concludes the proof.
\end{proof}

\Cref{obs:paths_aug} also allows us to show that \CMPshort is FPT when parameterized by $k$ and $\lambda$ by making use of the \emph{color coding} technique of Alon, Yuster, and Zwick~\cite{AlonYZ95}.

\begin{lemma}\label{lem:fptlambda}
\CMPshort can be solved in $2^{\Oh(k\lambda)}\cdot n^{\Oh(1)}$ time.     
\end{lemma}

\begin{proof}
The color coding technique can be used directly but, for simplicity, we reduce \CMPshort to the variant of \textsc{Disjoint Paths} where each path should have a prescribed length. This variant belongs to the class of problems solvable by color coding~\cite{AlonYZ95} (see also the textbook~\cite{CyganFKLMPPS15}). By~\Cref{obs:paths_aug}, \CMPshort is equivalent to \textsc{Disjoint Paths} on $\mathcal{G}_\lambda^*$ with the terminal pairs $(s_i^*,t_i^*)$ for~$i\in[k]$. Notice that for any~$i\in[k]$, the length of an $s_i^*$-$t_i^*$-path is at least 4 and at most~$2\lambda+2$. We guess the length $\ell_i\leq 2\lambda+2$ of the $s_i^*$-$t_i^*$-path in a (potential) solution. For each of at most $(2\lambda)^k$ guesses, we use the color coding algorithm from~\cite{AlonYZ95} to check whether  $\mathcal{G}_\lambda^*$ has vertex disjoint $s_i^*$-$t_i^*$-paths $P_i$ where the length of $P_i$ is exactly $\ell_i$. This can be done in~$2^{\Oh(\sum_{i=1}^k(\ell_i+1))}\cdot n^{\Oh(1)}$ time. If we find a solution for one of the guesses, we conclude that the considered instance of \CMPshort is a yes-instance. Otherwise, if we fail to find a solution for all guesses, we conclude that we have a no-instance. 
As we have at most $(2\lambda)^k$ guesses and $\sum_{i=1}^k(\ell_i+1)\leq (2\lambda+3)k$, we obtain that the overall running time is  $2^{\Oh(k\lambda)}\cdot n^{\Oh(1)}$, concluding the proof.
\end{proof}

\section{\CMPshort parameterized below $\upperboud$}
\label{sec:fpt}

In this section, we  prove our main result, \Cref{thm:fptbelow}:  \CMPshort is FPT when parameterized by $k$ and $\bellowPara=\upperboud_{\R, G}-\lambda$.
 Throughout the section, we assume that $\lambda\leq\upperboud_{\R,G}$ in the considered instances 
 (otherwise, we have a trivial yes-instance by \Cref{lem:upperbound}).
We start with some notation and auxiliary results.

\begin{figure}[t]
\centering
\scalebox{0.7}{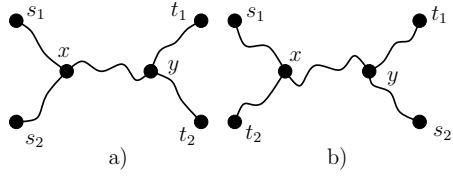}
\caption{Laminar codirectional (a) and opposite (b) $s_1$-$t_1$ and $s_2$-$t_2$-paths. }
\label{fig:laminar}
\end{figure}

Let $(s_1,t_1)$ and $(s_2,t_2)$ be pairs of distinct vertices of a graph $G$. We say that an $s_1$-$t_1$-path~$P_1$ and $s_2$-$t_2$-path $P_2$ are \emph{laminar} if either they have no common vertices or the intersection of $P_1$ and $P_2$ is a path in $G$ (see \Cref{fig:laminar}). Laminar paths $P_1$ and $P_2$ are \emph{codirectional}  if either they are disjoint or their intersection is an $x$-$y$-path $Q$ where either~${\dist_G(s_1,x)\leq \dist_G(s_1,y)}$ and 
${\dist_G(s_2,x)\leq \dist_G(s_2,y)}$ or  
${\dist_G(s_1,y)\leq \dist_G(s_1,x)}$ and~${\dist_G(s_2,y)\leq \dist_G(s_2,x)}$ (see \Cref{fig:laminar}(a)). Otherwise, we say that $P_1$ and~$P_2$ are \emph{opposite} (see \Cref{fig:laminar}(b)). We make the following easy observation.

\begin{observation}\label{obs:laminar}
For any two pairs $(s_1,t_1)$ and $(s_2,t_2)$ of distinct vertices of a graph $G$, $G$ has laminar shortest $s_1$-$t_1$- and $s_2$-$t_2$-paths. Moreover, for any given 
$s_2$-$t_2$-path
$P_2$, there is a shortest $s_1$-$t_1$-path $P_1'$ such that $P_1'$ and $P_2$ are laminar. \end{observation}

\begin{proof}
Consider arbitrary shortest $s_1$-$t_1$ and $s_2$-$t_2$-paths $P_1$ and $P_2$. If the paths are disjoint, then they are laminar by definition. Otherwise, let $x$ and $y$ be the vertices of $V(P_1)\cap V(P_2)$ such that the vertices of the $s_1$-$x$ and $y$-$t_2$-subpaths of $P_1$ (except for~$x$ and $y$) are not in $P_2$; note that it may happen that $x=s_1$ or $y=t_2$. Consider the $s_1$-$t_1$-path~$P_1'$ obtained by concatenating the $s_1$-$x$-subpath of $P_1$, the $x$-$y$-subpath of $P_2$, and the $y$-$t_1$-subpath of $P_1$. Because every subpath of a shortest path is a shortest path, we have that~$P_1'$ is a shortest $s_1$-$t_1$-path. By construction, $P_1'$ and $P_2$ are laminar. This concludes the proof.  
\end{proof}

In the series of the next lemmas, we give sufficient conditions ensuring that an instance of \CMPshort is a yes-instance. 
The first condition is the existence of two codirectional laminar shortest paths of length $\Omega(\bellowPara)$. 

\begin{lemma}\label{lem:codir}
Let $\mathcal{I}=(G,k,\R=\{R_i=(s_i,t_i) \mid i \in [k]\},\lambda)$ be an instance of \CMPshort, and let $\bellowPara=\upperboud_{\R, G}-\lambda$. Suppose that there are distinct $i,j\in[k]$ such that
\begin{itemize}
\item[(i)] $\dist_G(s_i,t_i)\geq \bellowPara+2$ and $\dist_G(s_j,t_j)\geq \bellowPara+2$,
\item[(ii)] there are codirectional laminar shortest $s_i$-$t_i$ and $s_j$-$t_j$-paths.   
\end{itemize}
Then, $\mathcal{I}$ is a yes-instance.
\end{lemma}

\begin{proof}
Let $P_i$ and $P_j$ be codirectional laminar shortest $s_i$-$t_i$ and $s_j$-$t_j$-paths,
and let $p=\dist_G(s_i,t_i)$, $q=\dist_G(s_j,t_j)$.
 
Assume first that $P_i$ and $P_j$ are disjoint. We assume without loss of generality that 
$p\leq q$. 
We define the routes of $R_i$ and $R_j$ as $(P_i,0)$ and $(P_j,0)$, respectively, that is, these agents move to their destinations along the shortest paths starting at the moment $0$. Since $P_i$ and $P_j$ are disjoint, the routes are not in conflict. They reach the destinations by the moment $q$. We route the remaining agents starting at the moment $q+1$. By \Cref{lem:upperbound}, they can be routed within the makespan at most 
$\sum_{h\in[k]\setminus\{i,j\}}\dist_G(s_h,t_h)+1$. Thus, we obtain a schedule with 
the total makespan at most
\begin{equation*}
\begin{aligned}
q+1+&\sum_{h\in[k]\setminus\{i,j\}}\dist_G(s_h,t_h)+1\\=&\sum_{h\in[k]\setminus\{i\}}\dist_G(s_h,t_h)+2\\
=&\sum_{h\in[k]}\dist_G(s_h,t_h)+2-p\\
=&\upperboud_{\R, G}+1-p.
\end{aligned}
\end{equation*}
Since $p\geq \bellowPara+2$, we have that the makespan of the schedule is at most
$\upperboud_{\R, G}-\bellowPara$. This means that $\mathcal{I}$ is a yes-instance. 

Suppose from now on that $P_i$ and $P_j$ have common vertices. Because $P_i$ and $P_j$ are laminar, their intersection is a path $Q$. Let $x$ and $y$ be the end-vertices of $Q$. Since the paths are codirectional, we can assume that
$\dist_G(s_i,x)\leq \dist_G(s_i,y)$ and $\dist_G(s_j,x)\leq \dist_G(s_j,y)$, that is, $x$ occurs first on the paths. Furthermore, we assume without loss of generality that
$\dist_G(s_i,x)\leq \dist_G(s_j,x)$.
We define the routes of $R_i$ and $R_j$ as $(P_i,0)$ and $(P_j,1)$, respectively, that is, these agents move to their destinations along the shortest paths, $R_i$ starts at the moment $0$, and $R_j$ starts one step later.
Because $\dist_G(s_i,x)\leq \dist_G(s_j,x)$, the routes are not conflicting. In particular, $R_j$ enters each vertex of the common part $Q$ after $R_i$ vacated it.   
Both $R_i$ and $R_j$ reach the destinations by the moment 
$\max\{p,q\}+1\leq p+q-(\bellowPara+2)+1=
p+q-\bellowPara-1$ as the length of $P_i$ and $P_j$ is at least $\bellowPara+2$.   
The remaining agents are routed starting at the moment $\max\{p,q\}+2\leq p+q-\bellowPara$ within the makespan 
$\sum_{h\in[k]\setminus\{i,j\}}\dist_G(s_h,t_h)+1$ using \Cref{lem:upperbound}. 
This gives a schedule with total makespan at most
\begin{equation*}
\begin{aligned}
p+q-\bellowPara+&\sum_{h\in[k]\setminus\{i,j\}}\dist_G(s_h,t_h)+1\\=&\sum_{h\in[k]}\dist_G(s_h,t_h)+1-\bellowPara\\
=&\upperboud_{\R, G}-\bellowPara.
\end{aligned}
\end{equation*}
We obtain that $\mathcal{I}$ is a yes-instance. This concludes the proof.
\end{proof}

The next sufficient condition is the existence of two opposite laminar shortest paths of length $\Omega(\bellowPara)$ such that one of the terminal vertices is sufficiently far from the common part of the paths.

\begin{lemma}\label{lem:opposite}
Let $\mathcal{I}=(G,k,\R=\{R_i=(s_i,t_i) \mid i \in [k]\},\lambda)$ be an instance of \CMPshort, and let $\bellowPara=\upperboud_{\R, G}-\lambda$. Suppose that there are distinct $i,j\in[k]$, such that
\begin{itemize}
\item[(i)] $\dist_G(s_i,t_i)\geq \bellowPara+2$ and $\dist_G(s_j,t_j)\geq \bellowPara+2$,
\item[(ii)] there are opposite laminar shortest $s_i$-$t_i$ and $s_j$-$t_j$-paths $P_i$ and $P_j$ with $Q=P_i\cap P_j$ being an $x$-$y$-path such that 
\begin{itemize}
\item $\dist_G(s_i,x)<\dist_G(s_i,y)$ and \\ $\dist_G(s_j,y)<\dist_G(s_j,x)$,
\item $\max\{\dist_G(s_i,x),\dist_G(t_i,y),\dist_G(s_j,y),\linebreak \dist_G(t_j,x)\}\geq \bellowPara+2$.
\end{itemize}
\end{itemize}
Then, $\mathcal{I}$ is a yes-instance.
\end{lemma}

\begin{proof}
By symmetry, assume that $\dist_G(s_i,x)\geq \bellowPara+2$; in particular, note that if we exchange the source and the target terminals for each agent then we obtain an equivalent instance.   Let  $P_i$ and $P_j$ be paths of length $p$ and $q$, respectively. 

We consider two cases. First, assume that 
$\dist_G(s_i,x)>\dist_G(s_j,x)$. 
We define the routes of $R_i$ and $R_j$ as $(P_i,0)$ and $(P_j,0)$. Because $\dist_G(s_i,x)>\dist_G(s_j,x)$, $R_i$ leaves the vertex $x$ before $R_j$ enters this vertex. This implies that the routes of $R_i$ and $R_j$ are not conflicting. These agents reach the destinations by the moment $\max\{p,q\}\leq p+q-(\bellowPara+2)$ because $p,q\geq \bellowPara+2$. 
The remaining agents are routed starting at the moment $\max\{p,q\}+1\leq p+q-\bellowPara$. By \Cref{lem:upperbound}, they can be routed within the makespan at most 
$\sum_{h\in[k]\setminus\{i,j\}}\dist_G(s_h,t_h)+1$. Then we obtain a schedule with 
the total makespan at most
\begin{equation*}
p+q-\bellowPara+\sum_{h\in[k]\setminus\{i,j\}}\dist_G(s_h,t_h)+1
= \upperboud_{\R,G}-\bellowPara
\end{equation*}
implying that $\mathcal{I}$ is a yes-instance. 

In the second case, $\dist_G(s_i,x)\leq \dist_G(s_j,x)$. For $R_j$, we use the route $(P_j,0)$, that is, this agent is moving from $s_j$ to $t_j$ along $P_j$ starting at zero time. To define the route for $R_i$, we set 
$\ell=\dist_G(s_j,x)-\dist_G(s_i,x)+1$. Then the route is $(P_i,\ell)$. By the choice of $\ell$, $R_i$ moves along $P_i$ and reaches $x$ after this vertex was vacated by $R_j$. Thus, the routes of $R_i$ and $R_j$ are not conflicting. The agent $R_j$ reaches its destination at the moment $q$, and $R_i$ -- at the moment
$\ell+1+p=\dist_G(s_j,x)-\dist_G(s_i,x)+1+p$. 
The remaining agents are routed starting at the moment $\max\{\ell+1+p,q\}+1$, and for their routing we use a schedule with the makespan at most  
$\sum_{h\in[k]\setminus\{i,j\}}\dist_G(s_h,t_h)+1$ using \Cref{lem:upperbound}.
Notice that $q\leq p+q-\bellowPara-1$ because $p\geq \bellowPara+2$, and
$\dist_G(s_j,x)-\dist_G(s_i,x)+1+p\leq q-\dist_G(s_i,x)+1+p\leq p+q-\bellowPara-1$. 
Then $\max\{\ell+1+p,q\}+1\leq p+q-\bellowPara$, and 
the total makespan is at most 
\begin{equation*}
p+q-\bellowPara+\sum_{h\in[k]\setminus\{i,j\}}\dist_G(s_h,t_h)+1
=\upperboud_{\R,G}-\bellowPara.
\end{equation*}
Thus, $\mathcal{I}$ is a yes-instance. This completes the proof.
\end{proof}

In the following lemma, we show that if there exist two pairs of terminals at distances $\Omega(\bellowPara)$ such that for one of them, there is a shortest path allowing to give way for the second agent, then we have a yes-instance.

\begin{lemma}\label{lem:middle}
Let $\mathcal{I}=(G,k,\R=\{R_i=(s_i,t_i) \mid i \in [k]\},\lambda)$ be an instance of \CMPshort, and let $\bellowPara=\upperboud_{\R, G}-\lambda$. Suppose that there are distinct $i,j\in[k]$ such that
\begin{itemize}
\item[(i)] $\dist_G(s_i,t_i)\geq 2\bellowPara+4$ and $\dist_G(s_j,t_j)\geq \bellowPara+3$,
\item[(ii)] there is a shortest $s_i$-$t_i$-path $P_i$ containing a vertex $w$ at distance at least $\bellowPara+2$ from both $s_i$ and $t_i$ whose degree in $G$ is at least three. 
\end{itemize}
Then $\mathcal{I}$ is a yes-instance.
\end{lemma}

\begin{proof}
Let $P_i=(u_0,\ldots,u_p)$. By \Cref{obs:laminar}, there is a shortest $s_j$-$t_j$-path $P_j=v_0,\ldots,v_q$ such that $P_i$ and $P_j$ are laminar. If $P_i$ and $P_j$ are codirectional then $\mathcal{I}$ is a yes-instance by \Cref{lem:codir}. Assume that the paths are opposite. Let $Q=P_1\cap P_2$ be an $x$-$y$-path and assume that 
$\dist_G(s_i,x)<\dist_G(s_i,y)$. If $w$ is not an internal vertex of $Q$ then because 
$\dist_G(s_i,w)\geq \bellowPara+2$ and $\dist_G(t_i,w)\geq \bellowPara+2$, we have that 
$\dist_G(s_i,x)\geq \bellowPara+2$ or $\dist_G(t_i,y)\geq \bellowPara+2$. Then $\mathcal{I}$ is a yes-instance by \Cref{lem:opposite}. From now on, we assume that $w$ is an internal vertex of $Q$. 
Since $w$ is an internal vertex of $Q$ and the degree of $w$ is at least three, $w$ has a neighbor $z$ in $G$ that is not in $Q$. Taking into account that $P_i$ and $P_j$ are shortest paths, we have that $z\notin V(P_i)$ and $z\notin V(P_j)$.     
By symmetry, we assume that $\dist_G(s_j,w)\geq \dist_G(t_j,w)$ as, otherwise, we can replace $(s_h,t_h)$ by $(t_h,s_h)$ for every $h\in[k]$ and obtain an equivalent instance. We consider two cases.

Assume that $\dist_G(s_i,w)\geq \dist_G(s_j,w)$. We define the route of $R_i$ as 
\[((u_0,\ldots,u_r,z,u_r,\ldots,u_p),0),\] where $u_r=w$. In words, the agent starts moving at the moment $0$ and moves along $P_i$ until it reaches $w$. Then the agent moves to $z$, returns to $w$, and follows $P_i$ up to the destination.  
To construct the route for $R_j$, we set 
$\ell=\dist_G(s_i,w)-\dist_G(s_j,w)+1\geq 0$. Then the route is defined as 
$(P_j,\ell)$, that is, $R_j$ moves along $P_j$ staring at the moment $\ell$.  
Note that $R_i$ reaches its destination at the moment $p+2$, and $R_j$ arrives to $t_j$ at the moment $\ell+q$. By the definition, $R_j$ reaches $w$ exactly at the same moment when $R_i$ moves to $z$. Thus, the routes are not conflicting. 
As with the proofs of the previous lemmas, the other agents are routed one step later after the arrival of $R_i$ and $R_j$, that is, at the moment $\max\{p+2,\ell+q\}+1$, and they are routed within the makespan 
at most $\sum_{h\in[k]\setminus\{i,j\}}\dist_G(s_h,t_h)+1$ by making use of \Cref{lem:upperbound}. 

To upper bound the total makespan, note
that since $q\geq \bellowPara+3$, $p+2\leq p+q-\bellowPara-1$. 
Because $\dist_G(t_i,w)\geq \bellowPara+2$, we have that 
\begin{equation*}
\ell+q=\dist_G(s_i,w)-\dist_G(s_j,w)+1+q\leq p-(\bellowPara+2)+1+q\leq p+q-\bellowPara-1. 
\end{equation*}
Then the makespan is at most
\begin{equation*}
p+q-\bellowPara+\sum_{h\in[k]\setminus\{i,j\}}\dist_G(s_h,t_h)+1
= \upperboud_{\R,G}-\bellowPara,
\end{equation*}
and $\mathcal{I}$ is a yes-instance. 

Suppose that $\dist_G(s_i,w)<\dist_G(s_j,w)$. The routes of $R_i$ and $R_j$ are defined symmetrically to the previous case. We set $\ell=\dist_G(s_j,w)-\dist_G(s_i,w)+1$ and define the route of $R_i$ as $(P_i,\ell)$. For $R_j$, we assume that $w=v_r$ and 
set the route to be \linebreak $((v_0,\ldots,v_r,z,v_r,\ldots,v_q),0)$. By the same arguments as in the first case, the routes are not in conflict. We have that $R_i$ reaches $t_i$ at the moment
$\ell+p$, and $R_j$ arrives to $t_j$ at the moment $q+2$. Then the other agents are routed staring at the moment $\max\{\ell+p,q+2\}+1$ within the makespan at most $\sum_{h\in[k]\setminus\{i,j\}}\dist_G(s_h,t_h)+1$ using \Cref{lem:upperbound}. 
As $\dist_G(s_i,w)\geq \bellowPara+2$, 
$$\ell+p=\dist_G(s_j,w)-\dist_G(s_i,w)+1+p\leq q-(\bellowPara+2)+1+p=p+q-\bellowPara-1.$$
We also have that
$q+2\leq p+q-\bellowPara-1$ because $p\geq 2\bellowPara+4$. This implies that the total makespan is at most 
$\upperboud_{\R,G}-\bellowPara$ as in the first case. Thus, $\mathcal{I}$ is a yes-instance.  
This concludes the proof.
\end{proof}

\Cref{lem:codir}--\Cref{lem:middle} are used to show the following three sufficient conditions that are crucial for our algorithm.
In \Cref{lem:threepaths}, we show that the existence of three pairs of terminals at distances $\Omega(\bellowPara)$ is a sufficient condition for having a yes-instance.

\begin{lemma}\label{lem:threepaths}
Let $\mathcal{I}=(G,k,\R=\{R_i=(s_i,t_i) \mid i \in [k]\},\lambda)$ be an instance of \CMPshort, and let $\bellowPara=\upperboud_{\R, G}-\lambda$. Suppose that there are distinct $h,i,j\in[k]$, such that
\begin{itemize}
\item[(i)] $\dist_G(s_h,t_h)\geq 2\bellowPara+3$,
\item[(ii)] $\dist_G(s_i,t_i)\geq \bellowPara+2$ and $\dist_G(s_j,t_j)\geq \bellowPara+2$. 
\end{itemize}
Then, $\mathcal{I}$ is a yes-instance.
\end{lemma}

\begin{proof}
Let $P_h$ be a shortest $s_h$-$t_h$-path. By \Cref{obs:laminar}, there are 
shortest $s_i$-$t_i$ and $s_j$-$t_j$-paths $P_i$ and $P_j$, respectively, such that $P_h,P_i$ and $P_h,P_j$ are laminar. If $P_h,P_i$ or $P_h,P_j$ are codirectional then $\mathcal{I}$ is a yes-instance by \Cref{lem:codir}. Assume that both $P_h,P_i$ and $P_h,P_j$ are opposite. Denote by $Q_1$ the $x_1$-$y_1$-path and by $Q_2$ the $x_2$-$y_2$-path that are common subpaths of $P_h,P_i$ and $P_h,P_j$, respectively. We assume that
$\dist_G(s_h,x_1)<\dist_G(s_h,y_1)$ and $\dist_G(s_h,x_2)<\dist_G(s_h,y_2)$.

If one of the distances $\dist_G(s_h,x_1)$, $\dist_G(s_h,x_2)$, $\dist_G(t_h,y_1)$, or $\dist_G(t_h,y_2)$ is at least $\bellowPara+2$ then $\mathcal{I}$ is a yes-instance by \Cref{lem:opposite}. Suppose that all these distances are at most $\bellowPara+1$. Let $x$ be the vertex of $P_h$ at distance $\max\{\dist_G(s_h,x_1),\dist_G(s_h,x_2)\}\leq \bellowPara+1$ from $s_h$, and let $y$ be the vertex at distance  $\max\{\dist_G(t_h,y_1),\dist_G(t_h,y_2)\}\leq \bellowPara+1$ from $t_h$. Since~$\dist_G(s_h,t_h)\geq 2\bellowPara+3$, we have that~$\dist_G(s_h,x)<\dist_G(s_h,y)$, that is, $x\neq y$. 
Note that $Q_1\cap Q_2$ have the common $x$-$y$-subpath $Q$, and 
$y$ is closer to $s_i$ and $s_j$ than $t_i$ and $t_j$, respectively.

Denote by $x'$ and $y'$ the vertices of $V(P_i)\cap V(P_j)$ such that the $t_i$-$x'$ and $s_i$-$y'$-subpaths of $P_i$ do not contain vertices of $P_j$ except $x'$ and $y'$, respectively. These vertices exists because $P_i$ and $P_j$ are not disjoint. 
Consider the $s_i$-$t_i$-path $P_i'$ obtained by concatenating the $s_i$-$y'$-subpath of $P_i$, the $y'$-$x'$-subpath of $P_j$, and the $x'$-$t_i$-subpath of $P_i$. Since  $P_i$ and $P_j$ are shortest path, we have that $P_i'$ and $P_j$ are laminar. 
Furthermore, because $P_i$ and $P_j$ are shortest paths and $s_i$ is closer to $y$ than to $x$, we have that $Q$ is a subpath of the $y'$-$x'$-subpath of $P_j$.
The paths~$P_i'$ and~$P_j$ are therefore codirectional. By \Cref{lem:codir}, we obtain that $\mathcal{I}$ is a yes-instance. This concludes the proof. 
\end{proof}

In \Cref{lem:onepath}, we give sufficient conditions for a yes-instance when there is a pair of terminals at sufficiently big distance.

\begin{lemma}\label{lem:onepath}
Let $\mathcal{I}=(G,k,\R=\{R_i=(s_i,t_i) \mid i \in [k]\},\lambda)$ be an instance of \CMPshort, and let $\bellowPara=\upperboud_{\R, G}-\lambda$. Let also $I=\{i\in[k]\mid \dist_G(s_i,t_i)\leq \bellowPara+1\}$. 
Suppose that 
\begin{itemize}
\item[(i)] there is an~$h\in[k]$ such that $5\bellowPara+4\leq\dist_G(s_h,t_h)\leq\upperboud_{\R,G}-\bellowPara$,
\item[(ii)] 
$\sum_{i\in I}\dist_G(s_i,t_i)\geq \bellowPara$.
\end{itemize}
Then, $\mathcal{I}$ is a yes-instance. Furthermore, if 
\begin{itemize}
\item[(a)] there is an~$h\in[k]$ such  that $3\bellowPara+2\leq\dist_G(s_h,t_h)$,
\item[(b)] either $|[k]\setminus I|=1$ and $\dist_G(s_h,t_h)\leq\upperboud_{\R,G}-\bellowPara$, or
$|[k]\setminus I|>1$ and $1+\sum_{i\in[k]\setminus I}\dist_G(s_i,t_i)\leq\upperboud_{\R,G}-\bellowPara$, and
\item[(c)] $\sum_{i\in I}\dist_G(s_i,t_i)<\bellowPara$,
\end{itemize}
then $\mathcal{I}$ is a yes-instance.
\end{lemma}

\begin{proof}
Both claims are shown by almost the same arguments. Thus, we prove the first claim and then explain how to modify the proof to show the second part of the lemma.

As the claim holds for $\bellowPara=0$ by \Cref{lem:upperbound}, we assume that $\bellowPara\geq 1$. 
We also assume without loss of generality that $5\bellowPara+4\leq\dist_G(s_1,t_1)$. Let $P_1=v_0,\ldots,v_p$ be a shortest $s_1$-$t_1$-path, and let $P_2,\ldots,P_k$ be shortest $s_i$-$t_i$-paths for $i\in[2,k]$, respectively. 

We select the set of indices $J\subseteq I$ as follows. If there is $i\in I$ such that $\dist_G(s_i,t_i)\geq \bellowPara$ then we set $J=\{i\}$ for arbitrary $i$ with this property. 
Otherwise, we greedily choose $J$ be an inclusion minimal set of indices such that
$\sum_{i\in J}\dist_G(s_i,t_i)\geq \bellowPara$. Notice that in this case,
$\sum_{i\in J}\dist_G(s_i,t_i)\leq 2\bellowPara-2$ since we have that $\dist_G(s_i,t_i)\leq \bellowPara-1$ for $i\in I$ when we make the greedy choice.

Then we construct the partition $\{J_1,J_2\}$ of $J$ (either $J_1$ or $J_2$ may be empty) where $J_1=\{i\in J\colon \text{ either }V(P_i)\cap V(P_1)=\emptyset\text{ or there is }j\in[0,2\bellowPara+1]\text { s.t. }v_j\in V(P_i)\}$ and $J_2=J\setminus J_1$. The condition that the paths $P_i$ for $i\in J$ are of length at most $\bellowPara+1$ implies the following claim.

\begin{claim}\label{cl:sep}
For every $i\in J_1$ and $j\in[3\bellowPara+3,p]$, $v_j\notin V(P_i)$, and for every 
$i\in J_2$ and $j\in[0,2\bellowPara+1]$, $v_j\notin V(P_i)$.
\end{claim}

\begin{proof}[Proof of \Cref{cl:sep}]
Notice that for every $i\in[0,2\bellowPara+1]$ and $j\in[3\bellowPara+3,p]$, $\dist_G(v_i,v_j)\geq \bellowPara+2$ as $P_1$ is a shortest path. Because for every $h\in[2,k]$, $\dist_G(s_h,t_h)\leq \bellowPara+1$, we have that there is no $P_h$ with $h\in[2,k]$ such that $P_h$ and $P_1$ have common vertices $v_i$ and $v_j$ for all  $i\in[0,2\bellowPara+1]$ and $j\in[3\bellowPara+3,p]$. This proves the claim
\end{proof}

Now we construct a schedule of makespan at most $\lambda$. 

First, we route the agents $R_i$ with $i\in [k]\setminus J$. Observe that these agents can be routed within the makespan $\lambda=\upperboud_{\R,G}-\bellowPara$ by \Cref{lem:upperbound} because $\sum_{i\in J}\dist_G(s_i,t_i)\geq \bellowPara$. 
Furthermore, they can be routed in such a way that ($*$)~the route of $R_1$ is along $P_1$, and ($**$)~whenever $R_1$ is in the internal vertices of $P_1$, no other vertex is occupied by any other agent.   
We consider such a schedule for $R_i$ with $i\in [k]\setminus J$, and assume that the route for $R_1$ is $(P_1,\ell)$. 

Next, we route the agents $R_i$ with $i\in J$. These agents are routed within the time interval $[\ell+1,\ell+5\bellowPara+3]$ using the paths $P_i$ for $i\in J$. Thus, because 
of properties~($*$) and ($**$) of the routing of the other agents and the condition that
$p\geq 5\bellowPara+4$, we only have to ensure that the routes of $R_i$ for $i\in J$ do not conflict with each other and the route of $R_1$.  

Consider the agents $R_i$ with $i\in J_2$. If $J_2=\emptyset$ then we do not route. If $|J_2|=1$ then the unique agent $R_i$ with $i\in J_2$
can be routed along $P_i$ within the makespan $\bellowPara+1\leq 2\bellowPara$ as $\dist_G(s_i,t_i)\leq \bellowPara+1\leq 2\bellowPara$. 
If $|J_2|\geq 2$ then $\sum_{i\in J_2}\dist_G(s_i,t_i)\leq\sum_{i\in I}\dist_G(s_i,t_i)\leq 2\bellowPara-2$, and the agents can be routed within the makespan $2\bellowPara-1<2\bellowPara$ by \Cref{lem:upperbound} via the paths $P_i$ for $i\in J_2$. In both cases, we can route the agents within the makespan $2\bellowPara$ using shortest paths avoiding conflicts between the routes. We use this schedule to route the agents with $i\in J_2$ starting at the moment $\ell+1$. Notice that the agents reach their destination by the moment $2\bellowPara+1$. By the choice of $J_2$ and \Cref{cl:sep}, we have that there are no conflicts between the agents $R_i$ for $i\in J_2$ and $R_1$.

The agents $R_i$ with $i\in J_1$ are routed symmetrically using the same arguments. They can be routed withing the makespan $2\bellowPara$ using $P_i$ for $i\in J_1$. We use this schedule to route them starting at the moment $\ell+3\bellowPara+3$. Then they reach their destination by the moment $\ell+5\bellowPara+3$. Therefore, we have that the routes do not conflict with each other and the route of $R_1$ by \Cref{cl:sep}. 

Since the agents $R_i$ with $i\in J_2$ and $R_j$ with $j\in J_1$ are routed within the time intervals $[\ell+1,\ell+2\bellowPara+1]$ and $[\ell+3\bellowPara+3,\ell+5\bellowPara+3]$, the routes of $R_i$ and $R_j$ are not conflicting. Thus, the constructed schedule is feasible. This proves the first claim.  

For the second part, we again assume that $\bellowPara\geq 1$, and because the claim is trivial for $k=1$, we let $k\geq 2$.
We also assume without loss of generality that $3\bellowPara+2\leq\dist_G(s_1,t_1)\leq\upperboud_{\R,G}-\bellowPara$. Let $P_1=v_0,\ldots,v_p$ be a shortest $s_1$-$t_1$-path, and let $P_2,\ldots,P_k$ be shortest $s_i$-$t_i$-paths for $i\in[2,k]$, respectively. 

Then we construct the partition $\{J_1,J_2\}$ of $I$ (either $J_1$ or $J_2$ may be empty) where $J_1=\{i\in I\colon \text{ either }V(P_i)\cap V(P_1)=\emptyset\text{ or there is }j\in[0,\bellowPara+1]\text { s.t. }v_j\in V(P_i)\}$ and $J_2=I\setminus J_1$. In the same way as with \Cref{cl:sep}, we have the following claim a the length of any $P_i$ for $i
\in[2,k]$ is at most $\bellowPara-1$. 

\begin{claim}\label{cl:septwo}
For every $i\in J_1$ and $j\in[2\bellowPara+1,p]$, $v_j\notin V(P_i)$, and for every 
$i\in J_2$ and $j\in[0,\bellowPara+1]$, $v_j\notin V(P_i)$.
\end{claim}

Then we construct a schedule of makespan at most $\lambda$. 

We route the agents $R_i$ with $i\in [k]\setminus I$. By \Cref{lem:upperbound} and condition (b) of the lemma, there is a routing of makespan at most  $\lambda=\upperboud_{\R,G}-\bellowPara$ such that ($*$)~the route of $R_1$ goes along $P_1$, and ($**$)~whenever $R_1$ is in the internal vertices of $P_1$, no other vertex is occupied by another agent. We assume that the route for $R_1$ is $(P_1,\ell)$. 

The agents $R_i$ with $i\in I$ are routed within the time interval $[\ell+1,\ell+3\bellowPara+1]$ using the paths $P_i$ for $i\in I$. As with the first part of the proof, because 
of properties~($*$) and ($**$) of the routing of the other agents and the condition that 
$p\geq 5\bellowPara+4$, we only have to ensure that the routes of $R_i$ for $i\in I$ do not conflict with each other and the route of $R_1$.  

First, we deal with the agents $R_i$ with $i\in J_2$. Since 
$\sum_{i\in J_2}\dist_G(s_i,t_i)\leq\sum_{i\in I}\dist_G(s_i,t_i)\leq \bellowPara-1$, the agents can be routed within the makespan $\bellowPara$ by \Cref{lem:upperbound} via the paths $P_i$ for $i\in J_2$.  This schedule is used to route the agents with $i\in J_2$ starting at the moment $\ell+1$. Notice that the agents reach their destination by the moment $\ell+\bellowPara+1$. By the choice of $J_2$ and \Cref{cl:septwo}, we have that there are no conflicts between the agents $R_i$ for $i\in J_2$ and $R_1$. The agents $R_i$ with $i\in J_1$ can be routed withing the makespan $\bellowPara$ using $P_i$ for $i\in J_1$, and we route them staring at the moment $\ell+2\bellowPara+1$. Then they reach their destination by the moment $\ell+3\bellowPara+1$. In the same way as before, the routes do not conflict with each other and the route of $R_1$ by \Cref{cl:septwo}. 

Since the agents $R_i$ with $i\in J_2$ and $R_j$ with $j\in J_1$ are routed within the time intervals $[\ell+1,\ell+\bellowPara+1]$ and $[\ell+2\bellowPara+1,\ell+3\bellowPara+1]$, the routes of $R_i$ and $R_j$ are not conflicting. This completes the proof.  
\end{proof}

In the following lemma, we provide sufficient conditions for a yes-instance when there are exactly two pairs of terminals at sufficiently big distances.

\begin{lemma}\label{lem:twopaths}
Let $\mathcal{I}=(G,k,\R=\{R_i=(s_i,t_i) \mid i \in [k]\},\lambda)$ be an instance of \CMPshort, and let $\bellowPara=\upperboud_{\R, G}-\lambda$. Suppose that there are distinct $i,j\in [k]$ such that 
\begin{itemize}
\item[(i)] $\dist_G(s_i,t_i)\geq 5\bellowPara+5$ and $\dist_G(s_j,t_j)\geq \bellowPara+2$,
\item[(ii)] the agents $R_i$ and $R_j$ can be routed within the makespan $\upperboud_{\R,G}-\bellowPara$ without conflicts with each other,
\item[(iii)] for every $h\in[k]$ such that $h\neq i,j$, $\dist_G(s_h,t_h)\leq \bellowPara+1$. 
\end{itemize}
Then, $\mathcal{I}$ is a yes-instance.
\end{lemma}

\begin{proof}
 The claim is trivial for $\bellowPara=0$ and it holds for $k=2$ by (ii). Thus, we assume that $\bellowPara\geq 1$ and $k\geq 3$. We also assume without loss of generality that $\dist_G(s_1,t_1)\geq 5\bellowPara+7$, $\dist_G(s_2,t_2)\geq \bellowPara+2$, that is, $i=1$ and $j=2$. 
If $\sum_{i=3}^k\dist_G(s_i,t_i)\geq \bellowPara$ then $\mathcal{I}$ is a yes-instance by \Cref{lem:onepath}. Therefore, we assume that $\sum_{i=3}^k\dist_G(s_i,t_i)<\bellowPara$. In particular, $\dist_G(s_i,t_i)\leq \bellowPara-1$ for $i\in [3,k]$.

By condition~(ii), $R_1$ and $R_2$ can be routed within the makespan $\upperboud_{\R,G}-\bellowPara$.
Denote by $r$ the minimum makespan for routing $R_1$ and $R_2$. Suppose that 
$(\upperboud_{\R,G}-\bellowPara)-r\geq \bellowPara+1$. Then we can route $R_1$ and $R_2$ using the optimal routing starting at the moment $0$. Since they would reach their destination by the moment $r$, the other agents can be routed staring at the moment $r+1$ using the schedule from \Cref{lem:upperbound}. Because $\sum_{i=3}^k\dist_G(s_i,t_i)\leq \bellowPara-1$, the total makespan of the schedule is at most $r+1+\bellowPara\leq \upperboud_{\R,G}$ implying that $\mathcal{I}$ is a yes-instance. By \Cref{lem:upperbound}, $r\leq\dist_G(s_1,t_1)+\dist_G(s_2,t_2)+1$, and by condition~(ii) $r\leq \upperboud_{\R,G}-\bellowPara$. Then if $r=\dist_G(s_1,t_1)+\dist_G(s_2,t_2)+1$, $\mathcal{I}$ is a yes-instance by \Cref{lem:onepath}. 
From now on, we assume that these are not the cases, and  
$\upperboud_{\R,G}-2\bellowPara\leq r\leq\dist_G(s_1,t_1)+\dist_G(s_2,t_2)$.
 
Let $P_1=(u_0,\ldots,u_p)$ and $P_2=(v_0,\ldots,v_q)$ be laminar $s_1$-$t_1$ and $s_2$-$t_2$-paths, respectively, such paths exist by \Cref{obs:laminar}. For $h\in[3,k]$, denote by $P_i$ a shortest $s_i$-$t_i$-path.

\begin{figure}[t]
\centering
\scalebox{0.7}{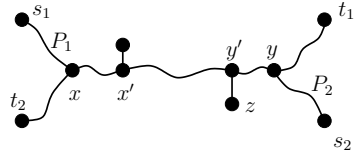}
\caption{The structure of $P_1$ and $P_2$. }
\label{fig:config}
\end{figure}

If $P_1$ and $P_2$ are codirectional then $\mathcal{I}$ is a yes-instance by \Cref{lem:codir}. Thus, we assume that they are opposite. Let $Q=P_1\cap P_2$ and denote by $x$ and $y$ its end-vertices assuming that $\dist_G(s_1,x)<\dist_G(s_1,y)$. If one of the distances $\dist_G(s_1,x)$, $\dist_G(t_2,x)$, $\dist_G(t_1,y)$, and $\dist_G(s_1,y)$ is at least $\bellowPara+2$ then $\mathcal{I}$ is a yes-instance by \Cref{lem:opposite}. Thus, we assume that these distances are at most $\bellowPara+1$. In particular, this means that 
$|\dist_G(s_1,t_1)-\dist_G(s_2,t_2)|\leq 2\bellowPara+2$ and, therefore, $\dist_G(s_2,t_2)\geq 3\bellowPara+3$.
If there is a vertex $w$ of degree at least three such that 
either $w\in V(P_1)$, $\dist_G(s_1,w)\geq \bellowPara+2$, and $\dist_G(t_1,w)\geq \bellowPara+2$, or 
$w\in V(P_2)$, $\dist_G(s_2,w)\geq \bellowPara+2$, and $\dist_G(t_2,w)\geq \bellowPara+2$, then $\mathcal{I}$ is a yes-instance by \Cref{lem:middle}. From now on, we assume that this in not the case. Then $Q$ contains an inclusion maximal $x'$-$y'$-subpath $Q'$ for some $x',y'\in V(Q)$ such that $\dist_G(s_1,x')\leq \bellowPara+1$,
$\dist_G(t_2,x')\leq \bellowPara+1$, $\dist_G(t_1,y')\leq \bellowPara+1$, $\dist_G(s_1,y')\leq \bellowPara+1$, and each internal vertex of $Q'$ has degree two in $G$. Notice that
$\dist_G(x,y)\geq\dist_G(x',y')\geq 3\bellowPara+3$. The obtained structure of $P_1$ and $P_2$ is shown in \Cref{fig:config}.

Recall that $R_1$ and $R_2$ can be routed with the makespan $r$. Denote by $(W_1,\ell_1)$ and $(W_2,\ell_2)$ the routes of $R_1$ and $R_2$, respectively, in such a schedule. 
We consider two cases depending on whether $R_1$ and $R_2$ traverse $Q'$ in this schedule. 

\subparagraph{Case~1.} For every edge~$e\in E(Q')$, both walks~$W_1$ and~$W_2$ traverse~$e$. 

\noindent To deal with this case, we use the following claim.

\begin{claim}\label{cl:route-one}
There is a schedule $\{(W_1',\ell_1'),(W_2',\ell_2')\}$ for $R_1$ and $R_2$ with the makespan $r$ such that either $\ell_1=0$, $W_1'=P_1$, and $\ell_2'\geq \dist_G(s_1,y')-\bellowPara-2$ or, symmetrically, $\ell_2=0$, $W_2'=P_2$, and $\ell_1'\geq \dist_G(s_2,x')-\bellowPara-2$.
\end{claim}

\begin{proof}[Proof of~\Cref{cl:route-one}]
Consider the walks $W_1$ and $W_2$. Because they traverse all edges of $Q'$, $W_1$ has a subwalk $A=(a_0,\ldots,a_f)$ such that $\{a_0,a_f\}=\{x',y'\}$ and $a_i\in V(Q')$ for $i\in[0,f]$, and $W_2$ has a subwalk $B=(b_0,\ldots,b_g)$ such that $\{b_0,b_g\}=\{x',y'\}$ and $b_i\in V(Q')$ for $i\in[0,g]$. 

For $W_1$, we claim that $a_0=x'$ and $a_f=y'$. For the sake of contradiction, assume that $W_1$ first visits $y'$ and then $x'$. Because $x'$ and $y$ are vertices of the shortest $s_1$-$t_1$-path $P_1$, we have that the length of $W_1$ is at least $\dist_G(s_1,t_1)+2\dist_G(x',y')$. Because $\dist_G(s_2,y'),\dist_G(t_2,x')\leq \bellowPara+1$, 
$\dist_G(s_2,t_2)\leq 2\bellowPara+2+\dist_G(x',y')<2\dist_G(x',y')$ as $\dist_G(x',y')\geq 3\bellowPara+3$. Thus, 
the length of $W_1$ is strictly more than $\dist_G(s_1,t_1)+\dist_G(s_2,t_2)$
contradicting that $r\leq \dist_G(s_1,t_1)+\dist_G(s_2,t_2)$. Thus, $W_1$ traverses $Q'$ from $x'$ to $y'$. Symmetrically, we have that $b_0=y'$ and $b_g=x'$, that is, $W_2$ goes from $y'$ to $x'$. 

Suppose that $R_1$ appears in $x'$ and start moving along $A$ before or at the same moment when $R_2$ appears in $y'$ and begins to follow $B$. Then because the internal vertices of $Q'$ have degree two in $G$, we have that $R_2$ can appear in $y'$ only when $R_1$ leaves this vertex. Thus, $R_2$ traverse $Q'$ after $R_1$. By the definition of $Q'$, either $t_1=y'$, or $s_2=y'$, or $y'$ has a neighbor $z$ in $G$ such that $z\notin V(P_1)$ (see \Cref{fig:config}). 

If $t_1=y'$, we set $W_1'=P_1$, $\ell_1'=0$, and we define $W_2'=P_2$ and $\ell_2'=\dist_G(s_1,y)-\dist_G(s_2,y')+1\geq \dist_G(s_1,y')-\bellowPara$. If $s_1=y'$, we let $W_1'=P_1$, $\ell_1'=0$, $W_2'=P_2$, and $\ell_2=\dist_G(s_1,y')+1\geq \dist_G(s_1,y')$.  In both cases, we have an optimal schedule because the routes are not conflicting,  
the agents follow shortest paths, and $R_2$ visits $y'$ immediately after $R_1$.

Suppose that $y'\neq t_1,s_2$. Then $y'$ has a neighbor $z\notin V(P_1)$. Recall that $P_2=(v_0,\ldots,v_q)$ and let $y'=v_h$ for $h\in[q]$. We define   
$W_1'=P_1$ and $\ell_1'=0$. For $R_2$, we set $W_2'=v_0,\ldots,v_h,z,v_h,\ldots,v_q$, that is, $R_2$ follows $P_2$ up to $y'$, then goes to $z$, returns to $y'$, and proceeds to $t_2$ by $P_2$, and set $\ell_2'=\dist_G(s_1,y')-\dist_G(s_2,y')-1\geq \dist_G(s_1,y')-\bellowPara-2$. As $z\notin V(P_1)$ and $R_2$ enters $y'$ after $R_1$, the routes do not conflict. Also, because $R_1$ follows a shortest path, and $R_2$ enters $y'$ immediately after this vertex is vacated by $R_1$ and then $R_2$ follows a shortest path to $t_2$, the schedule is optimal.  

The case when $R_2$ appears in $y'$ and start moving along $B$ before $R_1$ appears in $x'$ following $B$ is symmetric. By the same arguments as before, we have 
that we have an optimal schedule such that $\ell_2=0$, $W_2'=P_2$, and $\ell_1'\geq \dist_G(s_2,x')-\bellowPara-2$. This proves the claim.
\end{proof}

We use \Cref{cl:route-one} and assume that $R_1$ uses the route $(P_1=(u_0,\ldots,u_p),0)$ and $R_2$ starts its route $(W_2,\ell_2)$ at the moment $\ell_2\geq \dist_G(s_1,y')-\bellowPara-2$ (the other case is symmetric). Denote by $U$ the set of vertices visited by $W_2$ in the first $\bellowPara+2$ steps starting from $\ell_2$.
We route the agents $R_i$ for $i\geq 3$ within the time interval $[3\bellowPara+1]$. Notice that $\dist_G(s_1,t_1)\geq\dist_G(x',y')\geq 3\bellowPara+3$ and $3\bellowPara+1\leq p$. 

We construct the partition $\{J_1,J_2\}$ of $[3,k]$ (either $J_1$ or $J_2$ may be empty) 
where 
$J_1=\{i\in [3,k]\colon \text{there is }j\in[0,\bellowPara+1]\text { s.t. }u_j\in V(P_i)\}$ and $J_2=J\setminus J_1$. Then the following holds. 

\begin{claim}\label{cl:septhree}
For every $i\in J_1$ and $j\in[2\bellowPara+1,3\bellowPara+2]$,
$u_j\notin V(P_i)$,
and for every $i\in J_2$ and $j\in[0,\bellowPara+1]$, 
$u_j\notin V(P_i)$. Furthermore, for every $i\in J_1$, $V(P_i)\cap U=\emptyset$. 
\end{claim}

\begin{proof}
The first part holds by the same arguments as \Cref{cl:sep} because the length of each $P_i$ for $i\in[3,k]$ is at most $\bellowPara-1$. To see the second part, suppose that there is  
$i\in J_1$ and $w\in U$ such that $w\in V(P_i)$. Since $P_i$ contains a vertex $u_j$ for $j\in[0,\bellowPara+1]$ and $\dist_G(s_i,t_i)\leq \bellowPara-1$, we have that $\dist_G(x',u)\leq 2\bellowPara$. Since $u\in U$, $\dist_G(u,s_2)\leq \bellowPara+2$. Then $\dist_G(x',s_2)\leq 3\bellowPara+2$. This contradicts the fact that $\dist_G(x',s_2)\geq \dist_G(x',y')\geq 3\bellowPara+3$. This proves the claim.
\end{proof}

By \Cref{lem:upperbound}, the agents $R_i$ with $i\in J_2$ can be routed within the makespan $\bellowPara$ as $\sum_{i=3}^k\dist_G(s_i,t_i)\leq \bellowPara-1$. We route them in the time interval
$[1,\bellowPara+1]$. By \Cref{cl:septhree}, the paths $P_i$ for $i\in J_2$ have no conflicts with the route of $R_1$. Because $R_2$ does not start its route until the moment 
$\ell_2\geq \dist_G(s_1,y')-\bellowPara-2\geq 2\bellowPara+1$, the routes of $R_i$ for $i\in J_2$ has no conflicts with the route of $R_2$.

Then we route the agents $R_i$ with $i\in J_1$. Again, they can be routed within the makespan $\bellowPara$, and we route them in the time interval $[2\bellowPara+1,3\bellowPara+1]$. Note that the routes have no conflicts with the routes for $R_j$ with $j\in J_2$. By \Cref{cl:septhree}, we also have that the routes have no conflicts with the route of $R_1$. 
Because $\ell_2\geq 2\bellowPara+1$ and the agents $R_i$ for $i\in J_1$ reach their destinations by the moment $3\bellowPara+1$, any conflict between the routes of $R_i$ and $R_2$ can only occur within the first $\bellowPara+2$ steps of $R_2$ staring at the moment $\ell_2$. Thus, $R_2$ should be in a vertex of $U$.
However, $V(P_i)\cap U=\emptyset$ by \Cref{cl:septhree}. Therefore, there is no conflict between the routes of $R_i$ for $i\in J_1$ and $R_2$.

Summarizing, we obtain a schedule of the total makespan at most $r\leq \upperboud_{\R,G}-\bellowPara$. This concludes the analysis of the first case. 

\subparagraph{Case~2.} There is an edge $e\in E(Q')$ such that $W_1$ or $W_2$ do not traverse $e$. 

\noindent Denote by $G'$ the graph obtained from $G$ by the deletion of the internal vertices of $Q'$. Then we have the following schedule for $R_1$ and $R_2$.

\begin{figure}[t]
\centering
\scalebox{0.7}{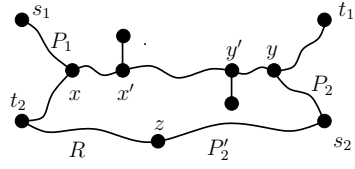}
\caption{The construction of  $P_2'$. }
\label{fig:rerout}
\end{figure}

\begin{claim}\label{cl:route-two}
There is a schedule $\{(P_1',0),(P_2',0)\}$ for $R_1$ and $R_2$ with the makespan $r$ such that either $P_1'=P_1$ and $P_2'$ is an arbitrary shortest $s_2$-$t_2$-path in $G'$ 
 or, symmetrically, $P_2'=P_2$ and $P_1'$ is an arbitrary shortest $s_1$-$t_1$-path in $G'$.
\end{claim}

\begin{proof}[Proof of \Cref{cl:route-two}]
Suppose that there is an edge of $Q'$ which is not traversed by $W_2$; the case when there is an edge of $Q'$ avoided by $W_1$ is symmetric. Then $s_2$ and $t_2$ are in the same connected component of $G'$. Let $P_2'$ be a shortest $s_2$-$t_2$-path (see~\Cref{fig:rerout}). Consider the routes $(P_1,0)$ and $(P_2',0)$ of $R_1$ and $R_2$, respectively. We claim that these routes do not conflict. By the definition of the routes, a possible conflict may happen only when $R_1$ goes through the $s_1$-$x'$ or $y'$-$s_2$-subpath 
of $P_1$. 
In the first case, $R_1$ should be at distance at most $\bellowPara$ from $x'$ in $G'$ as the conflict cannot happen at the moment $0$, and $R_1$ should be at distance at most $\bellowPara+1$ from $y'$ in the second. Note that when 
$R_1$ goes through the $s_1$-$x'$-subpath of $P_1$, $R_2$ is at distance at most 
$\dist_G(s_2,y')+\bellowPara+1\leq 2\bellowPara+2$ from $y'$. 
Since $\dist_{G'}(x',y')\geq \dist_G(x',y')\geq 3\bellowPara+3$, there is no conflict between $R_1$ and $R_2$. 
For the case when $R_1$ goes through the $y'$-$t_1$-subpath of $P_1$, note that 
because $\dist_G(x',y')\geq 3\bellowPara+3$, $R_2$ is at distance at least $3\bellowPara+3$ from 
$s_2$ in $G'$. Thus, $R_2$ is at distance at least $2\bellowPara+2$ from $y'$. This implies that we have no conflicts between the routes.

Clearly, the length of $P_1$ is at least $r$, and since 
$W_2$ does not traverse $Q'$, the length of $W_2$ is at least the length of $P_2'$.  
Therefore, the constructed schedule is optimal. This proves the claim.
\end{proof}

By using \Cref{cl:route-two}, we first assume that $R_1$ and $R_2$ use the routes $(P_1,0)$ and $(P_2',0)$ where $P_2'$ is a shortest $s_2$-$t_2$-path in $G'$. Recall that the minimum makespan $r\geq \upperboud_{\R,G}-2\bellowPara\geq \dist_G(s_1,t_1)+\dist_G(s_2,t_2)-2\bellowPara\geq 6\bellowPara+8$, that is, the length of $P_2'$ is at least $6\bellowPara+8$. Thus, $r$ is the length of $P_2'$.
Let $z$ be the vertex of $P_2'$ at distance $3\bellowPara+1$ from $t_2$ in $G'$, and denote by $R=(w_0,\ldots,w_{3\bellowPara+1})$ the $z$-$t_2$-subpath of $P_2'$ (see~\Cref{fig:rerout}). 

We define the partition $\{J_1,J_2\}$ of $[3,k]$ (either $J_1$ or $J_2$ may be empty) 
where 
$J_1=\{i\in [3,k]\colon \text{there is }j\in[2\bellowPara-1,3\bellowPara+1]\text { s.t. }w_j\in V(P_i)\}$ and $J_2=J\setminus J_1$. We have the following claim.

\begin{claim}\label{cl:sepfour}
For every $i\in J_1$ and $j\in[0,\bellowPara]$,
$w_j\notin V(P_i)$,
and for every $i\in J_2$ and $j\in[2\bellowPara-1,3\bellowPara+1]$, 
$w_j\notin V(P_i)$. 
Furthermore, for every $i\in J_1$ and every $j\in[p-2\bellowPara,p]$, $u_j\notin V(P_i)$.
\end{claim}

\begin{proof}

The first claim is fulfilled  by the same arguments as \Cref{cl:sep} as the length of each $P_i$ for $i\in[3,k]$ is at most $\bellowPara-1$. To see the second part, suppose that there is  
$i\in J_1$ and $j\in [p-2\bellowPara,p]$ such that $u_j\in V(P_i)$. 
Assume that $w_h\in V(P_1)$ for $h\in J_1$.
Notice that $w_h$ is at distance at most $\bellowPara+1$ from $t_2$ in $G$, $w_h$ is at distance at most $\bellowPara-1$ from 
$u_j$. Suppose that $u_j$ is at distance at most $\bellowPara+1$ from $y'$ in $G$. Then $t_2$ is at distance at most $3\bellowPara+1$ in $G$. However, this contradicts the fact that $\dist_G(t_2,y')\geq \dist_G(x',y')\geq 3\bellowPara+3$. Thus, $u_j$ is at distance at least $\bellowPara+2$ and at most $2\bellowPara$ from $y'$. Then $u_h$ is an internal vertex of $Q'$ as $y'$ is at distance at most $\bellowPara+1$ from $t_1$ and the length of $Q'$ is at least $3\bellowPara+1$. But in this case, we have that $P_i$ is a subpath of $Q'$ containing only internal vertices of $Q'$. Then $P_i$ cannot contain a vertex of $P'$.
This contradiction proves the claim.
\end{proof}

Now we can route the agents $R_i$ for $i\in[3,k]$ within the makespan $3\bellowPara$ in the time interval $[r-3\bellowPara-1,r-1]$. 

By~\Cref{lem:upperbound}, the agents $R_i$ for $i\in J_1$ can be routed within the makespan $\bellowPara$ because $\sum_{i=3}^k\dist_G(s_i,t_i)\leq \bellowPara-1$. We route them using this schedule starting at the moment $r-3\bellowPara-1$. By the first part of \Cref{cl:sepfour}, there is no conflict between $R_2$ and $R_i$ for every $i\in J_1$ because $r$ is the length of $P_2'$. Because $r\geq \dist_G(s_1,t_1)+\dist_G(s_2,t_2)-2\bellowPara\geq p+\bellowPara+3$,
$r-3\bellowPara-3\geq p-2\bellowPara$. Then by the second part of~\Cref{cl:sepfour}, there is no conflicts between $R_i$ for $i\in J_1$ and $R_1$.

Similarly to the agents $R_i$ with $i\in J_1$, the agents with $i
\in J_2$ can be routed within the makespan $\bellowPara$ by~\Cref{lem:upperbound}. We route them starting at the moment $r-\bellowPara-1$. Clearly, these agents have no conflicts with the agents $R_i$ for $i\in J_1$. Also, they have no conflicts with $R_2$ by the first part of \Cref{cl:sepfour}. Finally, because $r\geq p+\bellowPara+3$, $R_1$ is not in $G$ when the agents $R_i$ with $i\in J_2$ are routed. Therefore, we have no conflicts with $R_1$. 

Summarizing, we obtain a schedule with the makespan $r\leq\upperboud_{\R,G}-\bellowPara$. Thus, $\mathcal{I}$ is a yes-instance.

The case when $R_1$ and $R_2$ use the routes $(P_1',0)$ and $(P_2,0)$ where $P_1'$ is a shortest $s_1$-$t_1$-path in $G'$ is analyzed by the same arguments using symmetry. In fact, this case gets easier because we have that  
$r\geq \dist_G(s_1,t_1)+\dist_G(s_2,t_2)-2\bellowPara\geq q+3\bellowPara+5$. Thus, the agent $R_2$ following $P_2$ is going to leave the graph before the moment $r-3\bellowPara-1$. Because of this, we can route the agents $R_i$ for $i\in[3,k]$ within the time interval $[r-3\bellowPara-1,r-1]$ avoiding conflicts with $R_2$. The remaining arguments are the same as above. This concludes the proof.  
\end{proof}

Now we are ready to prove \Cref{thm:fptbelow}. We restate it here. 

\FPTBelowTheorem*

\begin{proof}
Let $\mathcal{I}=(G,k,\R=\{R_i=(s_i,t_i) \mid i \in [k]\},\lambda)$ be an instance of \CMPshort and let $\bellowPara=\upperboud_{\R, G}-\lambda$. If $\bellowPara=0$, then $\mathcal{I}$ is a yes-instance by \Cref{lem:upperbound}. Assume that $\bellowPara\geq 1$. We also assume that 
$\dist_G(s_1,t_1)\geq\dots\geq\dist_G(s_k,t_k)$.

If $\dist_G(s_1,t_1)>\lambda$, then we have a trivial no-instance.
If $k\geq 2$, we use the algorithm from \Cref{lem:XP} to check whether 
$(G,2,\{R_1,R_2\},\lambda)$ is a yes-instance. If $k=2$, then we solved the problem.
If $k\geq 3$ and we got the no-answer, then $\mathcal{I}$ is a no-instance.
Note that the algorithm runs in polynomial time for two agents.
From now on, we assume that~$k\geq 3$ and there is a schedule for $R_1$ and $R_2$ with makespan at most $\lambda$.

Next, we consider the cases depending on the distances between terminals and either use \Cref{lem:threepaths}, \Cref{lem:onepath}, \Cref{lem:twopaths} to argue that $\mathcal{I}$ is a yes-instance or apply \Cref{lem:fptlambda} to solve the problem.

Suppose first that $\dist_G(s_1,t_1)\leq 2\bellowPara+2$. Then, we have~$\sum_{i=1}^k\dist_G(s_i,t_i)\leq (2\bellowPara+2)k$ and~${\lambda\leq (2\bellowPara+2)k}$ as $\bellowPara\geq 1$. In this case,
we apply the algorithm from \Cref{lem:fptlambda} and solve the problem in $2^{\Oh(k^2\bellowPara)}\cdot n^{\Oh(1)}$ time. If~$\dist_G(s_1,t_1)\geq 2\bellowPara+3$ and $\dist_G(s_3,t_3)\geq \bellowPara+2$, then we conclude that $\mathcal{I}$ is a yes-instance by~\Cref{lem:threepaths}.  
Thus, from now on, we assume that~$\dist_G(s_i,t_i)\leq \bellowPara+1$ for each~$i\in [3,k]$. 

Suppose that $\dist_G(s_1,t_1)\leq 5\bellowPara+3$. Then,
\[\lambda\leq \sum_{i=1}^k\dist_G(s_i,t_i)\leq 2(5\bellowPara+3)+(k-2)(\bellowPara+1)= k(\bellowPara+1)+8\bellowPara+4.\]
We apply the algorithm from \Cref{lem:fptlambda} and solve the problem in $2^{\Oh(k^2\bellowPara)}\cdot n^{\Oh(1)}$ time.

From now on, assume that $\dist_G(s_1,t_1)\geq 5\bellowPara+4$. If~$\dist_G(s_2,t_2)\leq \bellowPara+1$, then 
$\mathcal{I}$ is a yes-instance by \Cref{lem:onepath}: if  $\sum_{i=2}^k\dist_G(s_i,t_i)\geq \bellowPara$ then we apply the first claim and use the second claim otherwise. Furthermore, if $\sum_{i=3}^k\dist_G(s_i,t_i)\geq \bellowPara$ then $\mathcal{I}$ is a yes-instance even if $\dist_G(s_2,t_2)\geq \bellowPara+2$ by the same lemma.

Now, we may assume~$\dist_G(s_2,t_2)\geq \bellowPara+2$ and~$\sum_{i=3}^k\dist_G(s_i,t_i)\leq \bellowPara-1$. If~$\dist_G(s_1,t_1)\geq 5\bellowPara+5$, then $\mathcal{I}$ is a yes-instance by \Cref{lem:twopaths}. Hence, it remains to consider the case when $\dist_G(s_1,t_1)\leq 5\bellowPara+4$. Here, we have that~$\lambda\leq\sum_{i=1}^k\dist_G(s_i,t_i)\leq 2(5\bellowPara+4)+\bellowPara-1=11\bellowPara+7$
and the problem can be solved in $2^{\Oh(k\bellowPara)}\cdot n^{\Oh(1)}$ time by \Cref{lem:fptlambda}. This concludes the case analysis.

Since the distances between terminals can be computed in polynomial time by standard algorithms~(see, e.g.,~\cite{CormenLRS09}), the overall running time is 
$2^{\Oh(k^2\bellowPara)}\cdot n^{\Oh(1)}$, concluding the proof.  
\end{proof}

 \section{The case of distinct terminals}\label{section:different-terminals}
In this section, we prove that \CMPshort is FPT when parameterized by only $\bellowPara=\upperboud_{\R, G}-\lambda$ if the terminals are pairwise distinct. The result is based on the following lemma.

\begin{lemma}\label{lem:three-dist}
Let $G$ be a connected graph and let $R_1=(s_1,t_1)$, $R_2=(s_2,t_2)$, and $R_3=(s_3,t_3)$ be agents such that the terminal vertices $s_1,s_2,s_3,t_1,t_2,t_3$ are distinct. Then there are distinct $i,j\in[3]$ such that the agents $R_i$ and $R_j$ can be routed with the makespan at most $\dist_G(s_i,t_i)+\dist_G(s_j,t_j)-1$ without conflicts.   
\end{lemma}

\begin{proof}
Assume that $\dist_G(s_1,t_1)\geq \dist_G(s_2,t_2),\dist_G(s_3,t_3)$, and let $P_1=(u_0,\ldots,u_p)$ be a shortest $s_1$-$t_1$-path. We select shortest $s_2$-$t_2$ and $s_3$-$t_3$-paths $P_2$ 
and 
$P_3$,
respectively, in such a way that $P_1,P_2$ and $P_1,P_3$ are laminar using \Cref{obs:laminar}. 

Suppose that there is $i\in[2,3]$ such that $P_1$ and $P_i$ are codirectional. We assume without loss of generality that $i=2$. 
If $P_1$ and $P_2$ have no common vertices then the routes $(P_1,0)$ and $(P_2,0)$ are not conflicting, and we have a routing with the makespan
$\dist_G(s_1,t_1)\leq \dist_G(s_1,t_1)+\dist_G(s_2,t_2)-1$.
Assume that $V(P_1)\cap V(P_2)\neq\emptyset$. Note that $\dist_G(s_1,t_1)\geq 2$ in this case. 
Because $P_1$ and $P_2$ are laminar, their intersection is an $x$-$y$-path $Q$ that is a common subpath of $P_1$ and $P_2$. We assume that $\dist_G(s_1,x)\leq \dist_G(s_1,y)$. 

Suppose that $\dist_G(s_1,x)\neq \dist_G(s_2,x)$. Then the routes $(P_1,0)$ and $(P_2,0)$ are not conflicting because the agent whose $s$-terminal is closer to $x$ reaches $x$ first and then follows the route. Thus, we can assume that $\dist_G(s_1,x)=\dist_G(s_2,x)$. Because $s_1\neq s_2$, we can route $R_1$ and $R_2$ using $(P_1,0)$ and $(P_2,1)$, respectively. Then because $R_2$ reaches $x$ after $R_1$, there are no conflicts. As the terminals are distinct and $\dist_G(s_1,t_1)\geq\dist_G(s_2,t_2)$, we have that $\dist_G(s_1,t_1)\geq 2$.
Then the total makespan is
$\max\{\dist_G(s_1,t_1),\dist_G(s_2,t_2)+1\}$.
 If   $\dist_G(s_1,t_1)>\dist_G(s_2,t_2)$ then 
the makespan is 
$\dist_G(s_1,t_1)\leq \dist_G(s_1,t_1)+\dist_G(s_2,t_2)-1$. If $\dist_G(s_1,t_1)=\dist_G(s_2,t_2)$ then  
the makespan is $\dist_G(s_2,t_2)+1\leq \dist_G(s_1,t_1)+\dist_G(s_2,t_2)-1$ as $\dist_G(s_1,t_1)\geq 2$.
In both cases, the makespan is upper bounded 
by $\dist_G(s_1,t_1)+\dist_G(s_2,t_2)-1$. Thus, the claim holds if $P_1$ and $P_2$ are codirectional.  

Suppose now that $P_1,P_2$ and $P_1,P_3$ are opposite. 
Consider $P_1$ and $P_2$, and let an $x$-$y$-path $Q$ be the common subpath of $P_1$ and $P_2$ where $\dist_G(s_1,x)<\dist_G(s_1,y)$.  

Assume that $x\notin \{s_1,t_2\}$. If $\dist_G(s_1,x)>\dist_G(s_2,x)$ then the routes $(P_1,0)$ and $(P_2,0)$ are not conflicting, and the makespan for this schedule is $\dist_G(s_1,t_1)\leq\dist_G(s_1,t_1)+\dist_G(s_2,t_2)-1$. Let $\dist_G(s_1,x)\leq\dist_G(s_2,x)$.
Consider the routes $(P_1,\dist_G(s_2,x)-\dist_G(s_1,x)+1)$ and $(P_2,0)$ of $R_1$ and $R_2$, respectively. Because $R_1$ enters $x$ after this vertex is vacated by $R_2$, the routes are not conflicting. Since $x\notin \{s_1,t_2\}$, $\dist_G(s_2,x)-\dist_G(s_1,x)+1\leq \dist_G(s_2,t_2)-1$.
As $\dist_G(s_1,t_1)\geq \dist_G(s_2,t_2)$, the makespan of the schedule is 
$\dist_G(s_1,t_1)+\dist_G(s_2,x)-\dist_G(s_1,x)+1\leq\dist_G(s_1,t_1)+\dist_G(s_2,t_2)-1$.

By symmetry, it remains to consider the case when $x\in\{s_1,t_2\}$ and $y\in\{s_2,t_1\}$. Because the terminals are distinct and $\dist_G(s_1,t_1)\geq \dist_G(s_2,t_2)$, we have that 
either $P_2=(u_{p-1},\ldots,u_1)$, or $P_2=(s_2,u_p,\ldots,u_1)$, or 
$P_2=u_{p-1},\ldots,u_0,t_2$. By the same arguments, either $P_3=(u_{p-1},\ldots,u_1)$, or $P_3=(s_3,u_p,\ldots,u_1)$, or 
$P_3=(u_{p-1},\ldots,u_0,t_3)$. Then, $P_2$ and $P_3$ are laminar codirectional paths. By the same arguments as in the case when $P_1$ is codirectional with $P_2$ or $P_3$, we have that $R_2$ and $R_3$ can be routed with the makespan $\dist_G(s_2,t_2)+\dist_G(s_3,t_3)-1$. This concludes the proof.
\end{proof}

By \Cref{lem:three-dist}, we can show the following.

\begin{lemma}\label{lem:upperbound-dist}
Let $(G, k, \R, \lambda)$ be an instance of \CMPshort such that the terminals are pairwise distinct.  
If $\lambda \geq \sum_{i\in[k]} \dist(s_i,t_i)-\big\lfloor\frac{k-1}{2}\big\rfloor$, then $(G, k, \R, \lambda)$ is a yes-instance.  
\end{lemma}

\begin{proof}
If $k=1$ then the claim is trivial. If $k=2$ then we route $R_1$ and $R_2$ using shortest~$s_1$-$t_1$- and $s_2$-$t_2$-paths~$P_1$ and $P_2$, respectively, using the routes 
$(P_1,0)$, $(P_2,\dist_G(s_1,t_1))$. Since~${t_1\neq s_2}$, there are no conflicts and this proves the claim. Assume that $k\geq 3$.

We greedily select $r=\big\lfloor\frac{k-1}{2}\big\rfloor$ pairs $\{i_1,j_1\},\ldots,\{i_r,j_r\}$ of indices such that all the indices are distinct and for every $h\in[r]$, the agents
$R_{i_h }$ and $R_{j_h}$ can be routed with the makespan at most $\dist_G(s_{i_h},t_{i_h})+\dist_G(s_{j_h},t_{j_h})-1$ without conflicts. This is possible due to \Cref{lem:three-dist}.  We fix these routes and concatenate them, that is $R_{i_1}$ and $R_{j_1}$ start at the moment $0$, and for each $h\in[2,r]$, $R_{i_h}$ and $R_{j_h}$ are routed 
when $R_{i_{h-1}}$ and $R_{j_{h-1}}$ reach their destinations. Then, we consecutively append the routes for $h\in[k]\setminus(\{i_1,\ldots,i_r\}\cup\{j_1,\ldots,j_r\})$ using shortest~$s_h$-$t_h$-paths. Because the terminals are distinct, we obtain a feasible schedule with the total makespan at most $\sum_{i\in[k]} \dist(s_i,t_i)-\big\lfloor\frac{k-1}{2}\big\rfloor$. This proves the lemma.
\end{proof}

We remark that \Cref{lem:upperbound-dist} implies that if the terminals are distinct then
$(G, k, \R, \lambda)$ is a yes-instance of \CMPshort if 
$\lambda \leq\upperboud_{\R,G}^* := \sum_{i\in[k]} \dist(s_i,t_i)$.

Combining \Cref{thm:fptbelow} and \Cref{lem:upperbound-dist}, we obtain the algorithmic result about \CMPshort parameterized below $\upperboud_{\R,G}^*$ for distinct terminals.

\FPTBelowTheoremTwo*

\begin{proof}
Let  $(G,k,\R=\{R_i=(s_i,t_i) \mid i \in [k]\},\lambda)$ be an instance of \CMPshort with pairwise distinct terminals and $\lambda=\upperboud_{\R,G}^*-\bellowPara$. If $\bellowPara\leq \big\lfloor\frac{k-1}{2}\big\rfloor$, this is a yes-instance by \Cref{lem:upperbound-dist}.
Otherwise, $k\leq 2\bellowPara$ and we can apply the algorithm from \Cref{thm:fptbelow} and solve the problem in~$2^{\Oh(\bellowPara^3)}\cdot n^{\Oh(1)}$~time.
\end{proof}

\section{Computational lower bounds}
\label{sec:hardness}
In this section, we prove our hardness results.
For this, it will be convenient to consider \textsc{Layered Vertex-Disjoint Shortest Paths}~\cite{BFG25}.
A graph is layered if there is a partition of the vertices into (disjoint) sets~$V_1,V_2,\ldots,V_{\lambda}$ such that all edges~$\{u,v\}$ are between two consecutive layers, that is~$u \in V_i$ and~${v \in V_{i+1}}$ or~$u \in V_{i+1}$ and~$v \in V_{i}$ for some~$i \in [\lambda-1]$.
An instance of \textsc{Layered Vertex-Disjoint Shortest Paths} consists of a layered graph and~$k$ terminal pairs~$(s_i,t_i)$ such that~$s_i \in V_1$, $t_i \in V_\lambda$, and each shortest~$s_i$-$t_i$-path contains exactly one vertex of each layer.
The question is whether there are~$k$ shortest paths (one between each terminal pair) that are pairwise vertex-disjoint.
It is known that \textsc{Layered Vertex-Disjoint Shortest Paths} does not admit a polynomial kernel parameterized by~$k+\lambda$ (the number of terminal pairs plus the number of layers) \cite{BFG25}.
Moreover, parameterized by only~$k$, the problem is known to be \classW1-hard \cite{Lochet21}.
We present a reduction from \textsc{Layered Vertex-Disjoint Shortest Paths} to \CMPlong{} based on a reduction by \cite{FioravantesKKMO24,FioKKMO25}.

\begin{theorem}
    \label{prop:W1k}
    \CMPlong{} is \classW1-hard parameterized by $k$.
\end{theorem}
\begin{proof}
    We reduce from \textsc{Layered Vertex-Disjoint Shortest Paths} parameterized by the number~$k'$ of terminal pairs, which is known to be \classW1-hard \cite{Lochet21}.
    Let~$\lambda'$ be the number of layers in the input graph.
    For each terminal pair~$(s_i,t_i)$, we construct an agent~$R_i=(s_i,t_i)$ and we set~$\lambda = \lambda'-1$.
    We next show that the reduction is correct.
    First assume that there is a set of~$k$ vertex-disjoint shortest paths.
    Then, we let each agent appear in time step~$0$, and let them move along the shortest path until reaching~$t_i$ in time step~$\lambda'-1=\lambda$.
    Since all paths are vertex-disjoint, this shows that the constructed instance is a yes-instance.
    In the other direction, assume that there is a schedule for all agents.
    Since the distance between~$s_i$ and~$t_i$ for each agent is by assumption~$\lambda'-1=\lambda$, the schedule has to let all agents appear at their respective starting vertices at time step~$0$ and let them move along their respective shortest paths towards their destinations.
    Hence, each agent is in layer~$i$ at time step~$i-1$ in the assumed schedule.
    This shows that the paths taken by all agents are pairwise vertex-disjoint and thus the original instance of \textsc{}\textsc{Layered Vertex-Disjoint Shortest Paths} is a yes-instance.
\end{proof}

We mention in passing that the above result can also be modified to show \classW1-hardness with respect to~$k$ even if the input graph is subcubic as shown by \cite{FioKKMO25}.
Simply replace each vertex by an in-tree followed by an out-tree (both of height $\log(n)$) and replacing each edge by an edge between a unique out-leaf of the vertex in the earlier layer and a unique in-leaf of the vertex in the latter layer.

\begin{corollary}\label{corrolary:w1-hard}
    \CMPlong{} is \classW1-hard when parameterized by~$k$ even if the input graph is subcubic.
\end{corollary}

We next show \classW1-hardness parameterized by  $\bellowPara$, the difference between the upper bound on the makespan and~$\lambda$.

\begin{restatable}{theorem}{Whard}
   \label{thm:W1b}
    \CMPshort{} is \classW1-hard parameterized by~$\bellowPara$.
\end{restatable}
\begin{proof}
    We reduce from \CMPlong{} parameterized by~$k$, which is \classW1-hard by \cref{prop:W1k}.
    Let~$(G,k,\mathcal{R},\lambda)$ be an input instance.
    We construct an equivalent instance~$(G',k',\mathcal{R}',\lambda')$ with~$k'=k+\lambda$ and~$\lambda'=\lambda$ as follows.
    First, we add to~$G$ a set of~$\lambda$ new vertices~$u_1,u_2,\ldots,u_{\lambda}$.
    Next, we add a new vertex~$v$ and make it adjacent to all other vertices in~$G$ (including the previously created vertices).
    We set~$\mathcal{R}'=\mathcal{R} \cup \{(v,u_i) \mid i \in [\lambda]\}$.
    Note that~$v$ ensures that the distance between each terminal pair in~$\mathcal{R}$ (in~$G'$) is at most~$2$ and the distance between all other terminal pairs is~$1$.
    Hence, $\bellowPara$ is at most~$1+2k+\lambda-\lambda=1+2k$.
    Since the reduction can clearly be computed in polynomial time, it only remains to show correctness.
    
    To this end, first assume that the original instance is a yes-instance.
    Then, we follow the same schedule for each agent in~$\mathcal{R}$.
    For each new agent~$(v,u_i)$, we let the agent appear at vertex~$v$ at time step~$i-1$ and move to vertex~$u_i$ at time step~$i$.
    Note that each agent reaches its destination latest at time step~$\lambda$ and no vertex is occupied by more than one agent at any time.
    Thus, the constructed instance is a yes-instance.

    In the other direction, assume that the constructed instance is a yes-instance.
    Then, each new agent~$(v,u_i)$ has to appear at vertex~$v$ at a distinct time step.
    Let~$(v,u_j)$ be the last agent appearing.
    Since~$u_j$ has to be reached by time step~$\lambda$, it holds that~$u_j$ appears at~$v$ at time step at most~$\lambda-1$.
    Hence, at each time step in~$[0,\lambda-1]$ a new agent appears at vertex~$v$.
    Thus, no agent in~$\mathcal{R}$ can occupy vertex~$v$ before time step~$\lambda$.
    Since~$v \neq t_i$ for each agent, no agent in~$\mathcal{R}$ occupies~$v$ at any time step.
    Thus, all agents follow a schedule in the original graph~$G$ and reach their destination by time step~$\lambda$.
    That is, the original instance of \CMPlong{} is a yes-instance, concluding the proof.
\end{proof}

We next prove \Cref{thm:nopolykernel}, that it is unlikely that \CMPlong{} has a polynomial kernel parameterized by~$k+\bellowPara$.
This result is shown combining ideas behind \cref{prop:W1k,thm:W1b}.

\begin{restatable}{theorem}{nokernel} \label{thm:nopolykernel}
    \CMPshort{} parameterized by~$k+\bellowPara$ does not admit a polynomial kernel unless 
    \classCoNP$\subseteq$\classNP/{\sf poly}.
\end{restatable}

\begin{proof}
    We design a polynomial parameter transformation from \textsc{Layered Vertex-Disjoint Shortest Paths} parameterized by~$k+\lambda$.
    This problem is known to not admit a polynomial kernel~\cite{BFG25}.
    Let~$(G,\{(s_i,t_i) \mid i \in [k]\})$ be an input instance of \textsc{Layered Vertex-Disjoint Shortest Paths} and let~$\lambda$ be the number of layers in~$G$.
    We construct an equivalent instance~$(G',k',\mathcal{R}',\lambda')$ with~$k'=k+\lambda-1$ and~$\lambda'=\lambda-1$ as follows.
    First, we add to~$G$ a set of~$\lambda-1$ new vertices~$u_1,u_2,\ldots,u_{\lambda-1}$.
    Next, we add a new vertex~$v$ and make it adjacent to all other vertices in~$G$ (including the previously created vertices).
    We set~$\mathcal{R}'=\{(s_i,t_i) \mid i \in [k]\} \cup \{(v,u_i) \mid i \in [\lambda-1]\}$.
    Note that~$v$ ensures that the distance between each terminal pair~$(s_i,t_i)$ is at most~$2$ (in~$G'$) and the distance between all other terminal pairs is~$1$ by construction.
    Hence, $\bellowPara$ is at most~$1+2k+\lambda-1-(\lambda-1)=1+2k$.
    Since the polynomial parameter transformation can clearly be computed in polynomial time, it only remains to show correctness.
    
    To this end, first assume that the original instance is a yes-instance.
    Let~$P_i$ be a path from~$s_i$ to~$t_i$ of length~$\lambda-1$ for each~$i \in [k]$.
    For each agent~$(s_i,t_i)$, we let the agent appear at~$s_i$ at time step~$0$ and then follow~$P_i$ so it reaches~$t_i$ at time step~$\lambda-1=\lambda'$.
    For each agent~$(v,u_i)$, we let them appear at~$v$ at time step~$i-1$ and move to~$u_i$ at time step~$i$.
    Note that each agent reaches its destination latest at time step~$\lambda'$ and no vertex is occupied by more than one agent at any time.
    Thus, the constructed instance is a yes-instance.

    In the other direction, assume that the constructed instance is a yes-instance.
    Then, each agent~$(v,u_i)$ has to appear at vertex~$v$ at a distinct time step.
    Let~$(v,u_j)$ be the last agent appearing.
    Since~$u_j$ has to be reached by time step~$\lambda'=\lambda-1$, it holds that~$u_j$ appears at~$v$ at time step at most~$\lambda-2$.
    Hence, at each time step in~$[0,\lambda-2]$ a new agent appears at vertex~$v$.
    Thus, no agent~$(s_i,t_i)$ can occupy vertex~$v$ before time step~$\lambda-1$.
    Since~$v \neq t_i$ for each agent, no agent in~$(s_i,t_i)$ occupies~$v$ at any time step.
    Thus, all these agents follow a schedule in the original graph~$G$ and reach their destination by time step~$\lambda'=\lambda-1$.
    Note that this implies that each such agent occupies a vertex of layer~$i$ at time step~$i-1$ for each~$i \in [\lambda]$.
    Thus, all these paths are vertex-disjoint showing that the original instance of \textsc{Layered Vertex-Disjoint Shortest Paths} is a yes-instance.
\end{proof}

 \section{Conclusion and open questions}\label{sec:conclusion}
In this work, we investigated \CMPshort{} through the lens of the above-and-below-guarantee paradigm of parameterized complexity. 
Our results provide a near-complete understanding of the problem with respect to the parameter~$k$ and~$\bellowPara$. 
We established fixed-parameter tractability for the combined parameterization and complement this with matching lower bounds showing that neither $k$ nor $\bellowPara$ alone is likely to yield an \classFPT algorithm. 
In addition, we ruled out the existence of a polynomial kernel for the combined parameter. 
Together, these results delineate the boundary between tractability and hardness for the problem in a rather precise manner.

We conclude with several open questions that we find particularly interesting. 

\begin{itemize}
    \item \textbf{Alternative upper bounds.} 
    From a combinatorial perspective, we were unable to prove or disprove whether the following inequality
    \[
        \lambda \leq 2 \cdot \max_{i \in [k]} \dist(s_i,t_i) + k
    \]
    holds for the minimum makespan $\lambda$ in general.  
    Establishing such a bound (or providing a counterexample) would be of independent interest. 
    A positive answer would, of course, give rise to new algorithmic questions within the above-and-below-guarantee framework.

    \item \textbf{Parameterization by $\bellowPara$.} 
    The precise complexity of \CMPshort{} when parameterized solely by $\bellowPara$ remains unresolved. 
    Is the problem in \classXP or para-\classNP-hard? 
    Clarifying this would complete the single-parameter landscape.

    \item \textbf{Restricted graph classes.} 
    Another promising direction is to study the complexity of the problem on restricted graph classes, such as planar graphs or graphs of bounded degree. 
    In particular, it remains open whether the problem is \classFPT when parameterized by $k$ on planar graphs or even on planar grids.
\end{itemize}

\end{document}